\documentclass[sn-mathphys-num]{sn-jnl}
\usepackage{graphicx} % Required for inserting images
\usepackage{stmaryrd}

\usepackage[ruled]{algorithm2e}

\usepackage{xcolor}
\definecolor{teal}{HTML}{008080}

\usepackage{hyperref}
\hypersetup{
   colorlinks=true,
   citecolor=teal   
}

\usepackage{amsmath}
\usepackage{amssymb, amsthm}

\newtheorem{proposition}{Proposition}
\newtheorem{theorem}{Theorem}

\newtheorem{definition}{Definition}

\allowdisplaybreaks[1] % To allow multiline equations over multiple pages. For amssmath package.

\usepackage{tikz}
\usetikzlibrary{quantikz2}
\usetikzlibrary{chains,positioning}
\tikzset{>=latex}

\newcommand{\descr}[1]{\vspace{0.2cm} \noindent \textbf{#1}}

\begin{document}

\title{Efficient Fault-Tolerant Quantum Protocol for Differential Privacy in the Shuffle Model
}

% \author{Hassan Asghar}
% \affiliation{%
%   \institution{School of Computing, Macquarie University}
%   \city{Sydney}
% %  \state{Ohio}
%   \country{Australia}}
% \email{hassan.asghar@mq.edu.au}

% \author{Arghya Mukherjee}
% \affiliation{%
%   \institution{School of Mathematical and Physical Sciences, Macquarie University}
%   \city{Sydney}
% %  \state{Ohio}
%   \country{Australia}}
% \email{arghya.mukherjee@students.mq.edu.au}

% \author{Gavin Brennen}
% \affiliation{%
%   \institution{School of Mathematical and Physical Sciences, Macquarie University}
%   \city{Sydney}
% %  \state{Ohio}
%   \country{Australia}}
% \email{gavin.brennen@mq.edu.au}

\author[1]{\fnm{Hassan} \sur{Asghar}}\email{hassan.asghar@mq.edu.au}

\author[2]{\fnm{Arghya } \sur{Mukherjee}}\email{arghya.mukherjee@students.mq.edu.au}

\author[2]{\fnm{Gavin K.} \sur{Brennen}}\email{gavin.brennen@mq.edu.au}

\affil*[1]{\orgdiv{School of Computing 
 }, \orgname{Macquarie
University},  \city{Sydney},  \state{NSW}, \country{Australia}}

\affil*[2]{\orgdiv{School of Mathematical and Physical
Sciences
 }, \orgname{Macquarie
University},  \city{Sydney},  \state{NSW}, \country{Australia}}

%\date{March 2023}

\abstract{
We present a quantum protocol which securely and implicitly implements a random shuffle to realize differential privacy in the shuffle model. The shuffle model of differential privacy amplifies privacy achievable via local differential privacy by randomly permuting the tuple of outcomes from data contributors. In practice, one needs to address how this shuffle is implemented. Examples include implementing the shuffle via mix-networks, or shuffling via a trusted third-party. These implementation specific issues raise non-trivial computational and trust requirements in a classical system. We propose a quantum version of the protocol using entanglement of quantum states and show that the shuffle can be implemented without these extra requirements. Our protocol implements $\kappa$-ary randomized response, for any value of $\kappa \ge 2$, and furthermore, can be efficiently implemented using fault-tolerant computation.}


\maketitle

\section{Introduction}
We consider the scenario of gathering data from remotely located individuals (clients), aggregating it and then releasing it in such a way that anyone receiving the processed data cannot learn anything specific about an individual, thus guaranteeing privacy of an individual. A concrete way of accomplishing this is via the notion of differential privacy~\cite{dwork2006calibrating}. More specifically, a central server, sometimes called an aggregator, receives individual datum from remote clients, applies an aggregation function (e.g., average), runs a differentially private mechanism on this output and then releases the result. Most commonly, the mechanism adds noise from a certain distribution given the desired privacy level, specified through the privacy parameter $\epsilon$, and the sensitivity of the function, which bounds how much the function can change if a single data item is to be added or removed. Intuitively, this protects privacy since the addition of noise masks any contribution from a specific individual. In most practical cases, individuals may not trust the server with their data, and hence could use the notion of local differential privacy (LDP)~\cite{kasiviswanathan2011ldp}. In this model, each client applies the differentially private mechanism directly on his/her input, usually referred to as a randomizer, and sends the perturbed input to the server. The server applies the aggregation function, and can optionally de-bias the result (which remains private due to the post-processing property of differential privacy~\cite{dwork2014dp-book}). More recently, another model known as the shuffle model has been proposed~\cite{bittau2017prochlo}. In this model, instead of sending the locally randomized inputs directly to the server, they are first randomly shuffled, after which the server is handed over all shuffled values at once. The idea is that the shuffling step, if performed securely, amplifies privacy, meaning that for the same privacy level, i.e., $\epsilon$, one gets better utility using the shuffle mechanism~\cite{privacy-blanket-shuffle}. The reason is intuitive: the server does not know which input belongs to which individual in the shuffle model.

In practice, one needs to determine how this shuffle mechanism is implemented. There are various methods. For instance, this could be done by a trusted third party~\cite{bittau2017prochlo} or via a mix network~\cite{dp-shuffle}, both residing between the clients and the server.\footnote{A slightly tangential approach is secure multiparty computation (SMC) through which clients can aggregate their own data and apply the differentially private mechanism using SMC techinques, thus simulating the central model~\cite{dwork2006ourdataourselves}. However, this solution incurs significant computation, bandwidth and liveness requirements on the clients~\cite{dp-shuffle}.} Inevitably, this adds further overhead to the protocol, as one needs to ensure that the shuffle is securely implemented. In this paper, we look at differential privacy in the shuffle model in the quantum world. More precisely, we assume that the remote clients as well as the server are equipped with quantum devices and connected via classical and quantum communication channels, and address the problem of differential privacy in the shuffle model. As we shall see, a characteristic of the application of the shuffle model in the quantum setting is that shuffling pretty much comes fror ``free,'' as we can utilize the \emph{entanglement} property of quantum states. We focus on the case where each device has classical information, which is then made differentially private by applying a local randomizer (which can be implemented classically or via a quantum circuit). Each client then performs a local measurement on the entangled state, and sends its result to the server via a classical channel. The proposed protocol implicitly implements the shuffle model mechanism where the local differential privacy algorithm is the $\kappa$-ary randomized response mechanism from~\cite{privacy-blanket-shuffle}. Our protocol, which is based on anonymous broadcasting, does not need a full purpose quantum computer to implement, since it only requires \emph{Clifford gates} (see Section~\ref{sec:fault-tolerance}). This is important because the overhead for quantum error correction to make the protocol robust in the presence of errors, is much less than for general purpose quantum computing. Furthermore, this allows us to easily implement our protocol in a fault-tolerant way which makes it suitable for the current generation of noisy intermediate scale quantum (NISQ) computers. To aid readers, our paper is mostly self-contained with almost all mathematical tools used for the quantum components in our construction introduced in this paper.

In what follows, we describe the threat model, and give a brief introduction to differential privacy and quantum computation in Section~\ref{sec:background}. In Section~\ref{sec:protocol} we describe our proposed protocol with proofs of correctness and security, and highlighting efficiency. Section~\ref{sec:circuits} describes all the quantum circuits used in our protocol and how they correctly compute the required outcome. We discuss fault-tolerant implementation of our protocol in Section~\ref{sec:fault-tolerance}. In Section~\ref{sec:physical} we discuss recent advances towards realizing the underlying qudit system used in our protocol. We discuss related work in Section~\ref{sec:rw} and present concluding remarks in Section~\ref{sec:conclude}. 
% Our main contributions are as follows.

% \begin{itemize}
%     \item First quantum protocol for differential privacy in the shuffle model
%     \item Our protocol is fault-tolerant unlike many secure quantum communication protocols
%     \item We provide a simulation 
% \end{itemize}

\section{Preliminaries and Background}
\label{sec:background}

\subsection{The Setting and Threat Model}
Our target setting is depicted in Figure~\ref{fig:mixnet-shuffle}. There are $n$ individuals, called clients, each with $x_i$. Each input is passed through a local randomizer $\mathcal{R}$ which outputs $y_i$ for input $x_i$. All outputs $y_i$ are shuffled, before sending them to the server. The server applies a function $f$ on the input. In the figure the shuffle model is depicted as being implemented as a mixnet. These notations are explained in more detail in the next section. We assume that the clients as well as the server are honest-but-curious and non-colluding. 

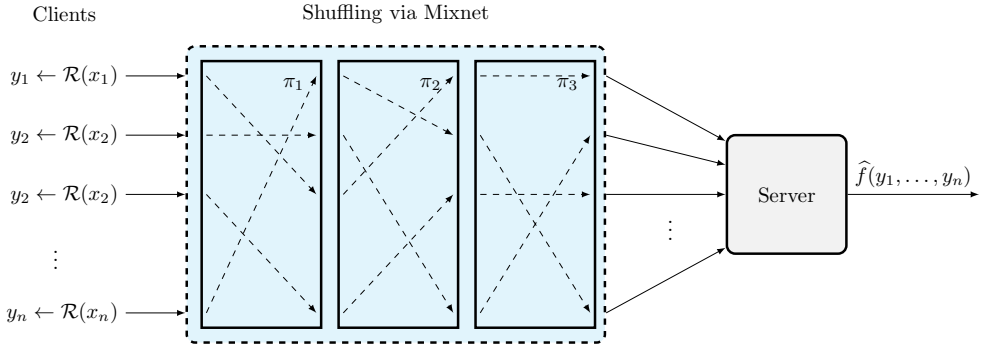
\begin{figure*}
    \centering
\resizebox{\columnwidth}{!}{
\begin{tikzpicture}
[
roundnode/.style={circle, draw=green!60, fill=green!5, very thick, minimum size=7mm},
squarednode/.style={rectangle, draw=red!60, fill=red!5, very thick, minimum size=5mm},
]
%Nodes
\node[minimum width = 7cm, minimum height = 5cm, draw=black, fill=cyan!10, dashed, rounded corners, very thick] 		(shuffler) 				{}; 

\node[] 		(shufflerlabel) 				[above =0.25cm of shuffler.north]{Shuffling via Mixnet}; 

\node[] 							(party1) 			[yshift=2cm, left=1cm of shuffler.west] {$y_1 \leftarrow \mathcal{R}(x_1)$};

\node[] 							(party2) 			[yshift=1cm, left=1cm of shuffler.west] {$y_2 \leftarrow \mathcal{R}(x_2)$};

\node[] 							(party3) 			[yshift=0cm, left=1cm of shuffler.west] {$y_2 \leftarrow \mathcal{R}(x_2)$};

\node[] 							(dots) 			[yshift=-1cm, left=2cm of shuffler.west] {$\vdots$};

\node[] 							(partyn) 			[yshift=-2cm, left=1cm of shuffler.west] {$y_n \leftarrow \mathcal{R}(x_n)$};

\node[] 		(partylabel) 				[above =0.5cm of party1.north]{Clients};

% mix net

\node[minimum width = 2cm, minimum height = 4.5cm, right=0.25cm of shuffler.west, draw=black, very thick] 		(mixnet1) 				{}; 

\node[] 		(mixnetlabel1) 				[below left =0.25cm of mixnet1.north east]{$\pi_1$}; 

\node[minimum width = 2cm, minimum height = 4.5cm, right=0.25cm of mixnet1.east, draw=black, very thick] 		(mixnet2) 				{}; 

\node[] 		(mixnetlabel2) 				[below left =0.25cm of mixnet2.north east]{$\pi_2$}; 

\node[minimum width = 2cm, minimum height = 4.5cm, right=0.25cm of mixnet2.east, draw=black, very thick] 		(mixnet3) 				{}; 

\node[] 		(mixnetlabel3) 				[below left =0.25cm of mixnet3.north east]{$\pi_3$}; 

% Arrows for Pi 1

\node[] 							(pi11start) 			[yshift=2cm, right=0.1cm of mixnet1.west] {};

\node[] 							(pi11end) 			[yshift=0cm, left=0.1 cm of mixnet1.east] {};

\draw[->, dashed] (pi11start.west) -- (pi11end.east);

\node[] 							(pi12start) 			[yshift=-2cm, right=0.1cm of mixnet1.west] {};

\node[] 							(pi12end) 			[yshift=2cm, left=0.1 cm of mixnet1.east] {};

\draw[->, dashed] (pi12start.west) -- (pi12end.east);

\node[] 							(pi13start) 			[yshift=0cm, right=0.1cm of mixnet1.west] {};

\node[] 							(pi13end) 			[yshift=-2cm, left=0.1 cm of mixnet1.east] {};

\draw[->, dashed] (pi13start.west) -- (pi13end.east);

\node[] 							(pi14start) 			[yshift=1cm, right=0.1cm of mixnet1.west] {};

\node[] 							(pi14end) 			[yshift=1cm, left=0.1 cm of mixnet1.east] {};

\draw[->, dashed] (pi14start.west) -- (pi14end.east);

% Arrows for Pi 2

\node[] 							(pi21start) 			[yshift=2cm, right=0.1cm of mixnet2.west] {};

\node[] 							(pi21end) 			[yshift=1cm, left=0.1 cm of mixnet2.east] {};

\draw[->, dashed] (pi21start.west) -- (pi21end.east);

\node[] 							(pi22start) 			[yshift=0cm, right=0.1cm of mixnet2.west] {};

\node[] 							(pi22end) 			[yshift=2cm, left=0.1 cm of mixnet2.east] {};

\draw[->, dashed] (pi22start.west) -- (pi22end.east);

\node[] 							(pi23start) 			[yshift=-2cm, right=0.1cm of mixnet2.west] {};

\node[] 							(pi23end) 			[yshift=0cm, left=0.1 cm of mixnet2.east] {};

\draw[->, dashed] (pi23start.west) -- (pi23end.east);

\node[] 							(pi24start) 			[yshift=1cm, right=0.1cm of mixnet2.west] {};

\node[] 							(pi24end) 			[yshift=-2cm, left=0.1 cm of mixnet2.east] {};

\draw[->, dashed] (pi24start.west) -- (pi24end.east);

% Arrows for Pi 3

\node[] 							(pi31start) 			[yshift=2cm, right=0.1cm of mixnet3.west] {};

\node[] 							(pi31end) 			[yshift=2cm, left=0.1 cm of mixnet3.east] {};

\draw[->, dashed] (pi31start.west) -- (pi31end.east);

\node[] 							(pi32start) 			[yshift=-2cm, right=0.1cm of mixnet3.west] {};

\node[] 							(pi32end) 			[yshift=1cm, left=0.1 cm of mixnet3.east] {};

\draw[->, dashed] (pi32start.west) -- (pi32end.east);

\node[] 							(pi33start) 			[yshift=0cm, right=0.1cm of mixnet3.west] {};

\node[] 							(pi33end) 			[yshift=0cm, left=0.1 cm of mixnet3.east] {};

\draw[->, dashed] (pi33start.west) -- (pi33end.east);

\node[] 							(pi34start) 			[yshift=1cm, right=0.1cm of mixnet3.west] {};

\node[] 							(pi34end) 			[yshift=-2cm, left=0.1 cm of mixnet3.east] {};

\draw[->, dashed] (pi34start.west) -- (pi34end.east);

% Aggregator

\node[minimum width = 2cm, minimum height = 2cm, right=2cm of shuffler.east, draw=black, very thick, rounded corners, fill=gray!10] 		(agg) 				{Server}; 

\draw[->] ([yshift=2cm]shuffler.east) -- ([yshift=0.9cm]agg.west);

\draw[->] ([yshift=1cm]shuffler.east) -- ([yshift=0.5cm]agg.west);

\draw[->] ([yshift=0cm]shuffler.east) -- ([yshift=0cm]agg.west);

\node[] 							(dots) 			[yshift=-0.5cm, left=0.75cm of agg.west] {$\vdots$};

\draw[->] ([yshift=-2cm]shuffler.east) -- ([yshift=-0.9cm]agg.west);

\draw[->] (agg.east) -- node [above, pos = 0.5] {$\widehat{f}(y_1, \ldots, y_n)$} ([xshift=2.2cm]agg.east);
 
%Lines
\draw[->] (party1.east) -- node [above, pos = 0.5] {} ([yshift=2cm]shuffler.west);
\draw[->] (party2.east) -- node [below, pos = 0.5] {} ([yshift=1cm]shuffler.west);
\draw[->] (party3.east) -- node [below, pos = 0.5] {} ([yshift=0cm]shuffler.west);
\draw[->] (partyn.east) -- node [below, pos = 0.5] {} ([yshift=-2cm]shuffler.west);

\end{tikzpicture}
}
\caption{The shuffle model of differential privacy. Each client's input $x_i$ is locally randomized, before being shuffled. Shuffling in this case is implemented via a mix network. The server then combines the shuffled and locally randomized values to produce a differentially private estimate $\widehat{f}$ of the function $f$ of the original inputs $x_1, \ldots, x_n$.}
\label{fig:mixnet-shuffle}
\end{figure*}

\subsection{Differential Privacy and the Shuffle Model}
\label{subsec:dp}
We summarize the local, central and shuffle model of differential privacy based on the formulation in~\cite{privacy-blanket-shuffle}. Let $\mathcal{X}$ be a data domain. We denote datasets as tuples $D = (x_1, \ldots, x_n) \in \mathcal{X}^n$, where each $x_i \in \mathcal{X}$. Two datasets $D, D'$ are neighbours, denoted $D \sim D'$ if they differ in exactly one element. Let $\epsilon \ge 0$, and $\delta \in [0, 1]$. 

\begin{definition}[Differential Privacy]
\label{def:dp}
A randomized algorithm $\mathcal{M}: \mathcal{X}^n \rightarrow \mathcal{Y}$ is $(\epsilon, \delta)$-differentially private (DP) if for all pairs of neighbouring datasets $D, D'$ and for all subsets $S \subseteq \mathcal{Y}$, we have:
\[
\Pr[\mathcal{M}(D) \in S] \leq e^\epsilon \Pr[\mathcal{M}(D') \in S] + \delta.
\]
\end{definition}
The following are known results.
\begin{proposition}[Post-processing~\cite{dwork2006calibrating}]\label{prop:post-processing}
    If $\mathcal{M}$ is $(\epsilon,\delta)$-DP, then for any algorithm $\mathcal{M}'$, $\mathcal{M}' \circ \mathcal{M}$ is also $(\epsilon,\delta)$-DP.
\end{proposition}

\begin{proposition}[Sequential composition~\cite{dwork2014dp-book}]\label{prop:composition}
    If $\mathcal{M}_1,\ldots,\mathcal{M}_t$ are $(\epsilon,\delta)$-DP, then the sequence of algorithms $\mathcal{M}' = \left(\mathcal{M}_1,\ldots,\mathcal{M}_t\right)$ is $(t\epsilon, t\delta)$-DP.
\end{proposition}

\begin{definition}[Local Differential Privacy]
A local randomizer is a randomized algorithm $\mathcal{R}: \mathcal{X} \rightarrow \mathcal{Y}$. We say that the local randomizer is $(\epsilon, \delta)$-locally differentially private (LDP) if for all $x, x' \in \mathcal{X}$ and for all $S \subseteq \mathcal{Y}$, we have:
\[
\Pr[\mathcal{R}(x) \in S] \leq e^\epsilon \Pr[\mathcal{R}(x') \in S] + \delta.
\]    
\end{definition}

If $\mathcal{R}$ is $(\epsilon, \delta)$-LDP, then the mechanism $\mathcal{M}: \mathcal{X}^n \rightarrow \mathcal{Y}^n$ defined as $\mathcal{M}(x_1, \ldots, x_n) = (\mathcal{R}(x_1), \ldots, \mathcal{R}(x_n))$ is $(\epsilon, \delta)$-DP. This mechanism provides privacy against the server as well. However, it incurs more accuracy loss than the central mechanism (Definition~\ref{def:dp}). The (single-message) shuffle model of differential privacy employs a shuffler $\mathcal{S}: \mathcal{Y}^n \rightarrow \mathcal{Y}^n$ which is a random permutation of its inputs. The algorithm $\mathcal{M}: \mathcal{S} \circ \mathcal{R}^n : \mathcal{X}^n \rightarrow \mathcal{Y}^n$ then provides $(\epsilon, \delta)$-DP against the server, but with the advantage that the local randomizer need only be $\epsilon_0$-LDP, with $\epsilon_0$ greater than $\epsilon$. Hence the gathered inputs are less noisy, resulting in better accuracy for the same value of $\epsilon$ as compared to a purely LDP algorithm. 

\descr{Local Differential Privacy Algorithm.} We use the local differential privacy (LDP) algorithm, i.e., randomizer, from~\cite{privacy-blanket-shuffle}, shown below. This is a $\kappa$-ary randomized response algorithm, where $\kappa \ge 2$. With $\kappa=2$, we get the binary randomized response. The goal of the server is to output the sum of true values. While we use this algorithm for our differentially private protocol, this can be replaced by any other algorithm whose the output lies in a discrete domain, e.g., the randomizer from~\cite{privacy-blanket-shuffle} for releasing the sum of normalized real numbered values to a given precision. 

% Without loss of generality, we assume that the input $x \in [0, 1]$. For a precision level $k$, we first encode $x$ as an integer as follows~\cite{privacy-blanket-shuffle}:
% \[
% \overline{x} = \lfloor xk \rfloor + \text{Ber}(xk - \lfloor xk \rfloor)
% \]
% It can be easily verified that encoding in this way ensures that $\mathbb{E}(\overline{x}/k)$, which is the expected value of the decoded $\overline{x}$, is exactly $\mathbb{E}(x)$. This makes the range of $\overline{x}$ equal to $\{0, 1, \ldots, k\}$. After this pre-processing step, we can use the local differential privacy algorithm from \cite{privacy-blanket-shuffle}:

\begin{algorithm}[!ht]
\SetAlgoLined
%\SetAlCapSkip{1em}
\DontPrintSemicolon{}
%\let\oldnl\nl% Store \nl in \oldnl
%\newcommand{\nonl}{\renewcommand{\nl}{\let\nl\oldnl}}% Remove line number for one line
%\nonl\TitleOfAlgo{{Select Features}}
\SetKwInOut{Input}{Input}
\SetKwInOut{State}{State}
\Input{$\kappa \in \mathbb{N}, {x} \in \{0, 1, \ldots, \kappa-1\}, \gamma \in [0, 1]$.}
%\State{A noise dictionary, denoted $\mathsf{nd}$, with keys from subsets of $U$ and values in $\mathbb{Z}_{\pm r}$.}
$b \leftarrow \text{Ber}(\gamma)$\;
\eIf {$b = 0 $}{
	$y \leftarrow {x}$\;
	}{
        $y \overset{{\scriptscriptstyle\$}}{\leftarrow} \{0, 1, \ldots, \kappa-1\}$\;
        }
return $y$\;
\caption{Local Differential Privacy Algorithm for Sums~\cite{privacy-blanket-shuffle}}
\label{algo:bounded-noise}
\end{algorithm}

This algorithm is $\epsilon$-differentially private with
% , as long as:
% \[
% \frac{1 - \gamma(\kappa-1)/\kappa}{\gamma/\kappa} \leq e^{\epsilon}
% \]
% Equating the left hand side to right hand side, we get:
\[
\gamma = \frac{\kappa}{\kappa - 1 + e^{\epsilon}}.
\]
%Thus, we can set $\gamma$ to this value given $\epsilon$ and $\kappa$. 

\descr{The Server.} The server sums all values $y_i$ and outputs the de-biased sum~\cite{privacy-blanket-shuffle} as:
\begin{equation}
\label{eq:de-bias}
    \frac{1}{1 - \gamma}\left( {\sum_{i=1}^n y_i} - \frac{\gamma (\kappa-1)n}{2} \right)
\end{equation}
For completeness, we show the derivation of this quantity in Appendix~\ref{app:de-bias}.

\descr{Privacy Amplification via Shuffling.} If the inputs from all clients are randomly shuffled, i.e., permuted via a permutation chosen uniformly at random, the resulting protocol amplifies the privacy of the standalone LDP mechanism. This means that we can use higher values of $\epsilon_0$. 
%Ignoring logarithmic terms, $\epsilon_0$ is proportional to $n$ and inversely proportional to $\delta$. 
Given a value of $\epsilon, \delta$ and $n$, we can use the script provided in~\cite{privacy-blanket-shuffle} to obtain a value of $\epsilon_0$ 
%which uses a tighter analysis than given by the implicit bounds in the paper. 
For instance, for the LDP mechanism described above with $\kappa=10$, with $n = 100$ participants, $\delta = 10^{-6}$ and $\epsilon = 0.1$, we get $\epsilon_0 \approx 1.0032$ through the Bennett bound from this script. This means, we can use the mechanism 10 times more then the LDP mechanism alone.

%\descr{Utility.} Involves a debiasing step.

\descr{Implementing the Shuffle.} Literature discusses multiple ways to implement the shuffle. One way is via mix-networks~\cite{dp-shuffle}, as shown in Figure~\ref{fig:mixnet-shuffle}, another is through a trusted third party~\cite{bittau2017prochlo}. In either case, we involve another party with the assumption that it does not collude with the server. In practice, this may also introduce additional overhead as the mix-network itself may be implemented using several servers relaying the message from one to another. As we shall see, using the properties of entanglement we do not need a third party for the shuffle in the quantum variant of the protocol.

\subsection{Background in Quantum Computation}
Let $d \geq 2$. Let $\mathbb{Z}_d = \{0, 1, \ldots, d - 1\}$, where addition of elements of $\mathbb{Z}_d$ is assumed to be done modulo $d$. We work in the $d$-dimensional complex Hilbert space $\mathbb{C}^d$~\cite{havlicek2007gen-pauli}, which for our purposes is simply a vector space endowed with an inner product. We call this the state space. A column vector from $\mathbb{C}^d$ is denoted by $\ket{v}$ (pronounced ``ket'' $v$). We denote the standard computational basis of $\mathbb{C}^d$ by $\{ \ket{s} : s \in \mathbb{Z}_d\}$, where $\ket{s}$ is the column vector with all zeros except at position $s$, where it is 1. Operations on vectors are defined by linear operators (matrices) $A$. We will exclusively consider matrices that map vectors from $\mathbb{C}^d$ to $\mathbb{C}^d$. Thus, $A$ is a $d \times d$ square matrix with elements from $\mathbb{C}^d$. For a matrix $A$, let $A^\dagger$ denote its Hermitian conjugate, obtained by taking the complex conjugate of each element and then transposing the matrix. A matrix is \emph{normal} if $A A^\dagger = A^\dagger A$. A normal matrix is \emph{unitary} if $A A^\dagger = A^\dagger A = I$. Note that we define $\ket{\psi}^\dagger$ as $\bra{\psi}$, which is a row vector with each entry being the complex conjugate of the corresponding entry in $\ket{\psi}$. With this notation, the dot product between two vectors $\ket{v}$ and $\ket{w}$ is defined as $\braket{v}{w}$ (a complex number), and their cross product as $\ket{v}\bra{w}$ (a $d \times d$ matrix). It can easily be verified that $\braket{v}{w}  = \braket{w}{v}^*$. We define a \emph{qudit} $\ket{\psi}$ as a vector of unit norm, i.e., $\sqrt{\braket{\psi}{\psi}} = 1$. We shall also call it a state vector. With this notation, note that for any two computational basis vectors $\ket{s}$ and $\ket{r}$, we have $\braket{s}{r} = \delta_{r, s}$, where $\delta_{r, s} = 1$ if $r = s$, and $0$ otherwise. The evolution of the state is described via unitary transformations, i.e., applying unitary matrices to states. These can be implemented as gates in a quantum circuit.

\descr{Primitive $d$th Roots of Unity.} Before we describe the unitary matrices used in our protocol, we introduce some facts about the $d$th roots of unity. Let $d \geq 2$ be an integer, and let $\omega = e^{i2\pi/d}$. We shall use the following fact about the $d$th roots of unity: $\{1, \omega, \omega^2, \ldots, \omega^{d-1}\}$ which forms an Abelian group under multiplication.
\begin{proposition}
\label{prop:unity}
Let $\omega = e^{i2\pi/d}$, where $d \geq 2$ is an integer. Let $x$ be an integer. Then (a) $\omega^x = \omega^{x \pmod{d}}$. Furthermore (b):
\begin{equation}
\label{eq:sum-unity}
\sum_{j = 0}^{d-1} (\omega^j)^x =
    \begin{cases}
    0, \text{ if } x \neq 0,\\
    d, \text{ if } x = 0
\end{cases}
\end{equation}
\end{proposition}
\begin{proof}
For part (a) let $x \equiv y \pmod{d}$, where $0 \leq y \leq d - 1$. Then there exists an integer $t$ such that:
\[
x = td + y
\]
Raising $\omega$ to this power we see that $\omega^x = (\omega^d)^t \omega^y = \omega^y$, where we have used the fact that $\omega^d = e^{i2\pi} = 1$. 

For part (b), through the sum of first $d$ terms of a geometric series, we get:
\[
\frac{1 - (\omega^x)^d}{1 - (\omega^x)} = \frac{1 - (\omega^d)^x}{1 - (\omega^x)} = 0,
\]
if $\omega^x \neq 1$. We see that $\omega^x = 1$ when $ e^{i2\pi x/d} = 1$, which is only possible if $x/d$ is an integer. From part (a), we may assume $x \in \mathbb{Z}_d$. Therefore, this is only possible if $x = 0$. When $x = 0$, it is straightforward to see that the sum is $d$.    
\end{proof}

\descr{Single Qudit Operations.} The (generalized) $X$ and $Z$ operators are defined as:
\[
X\ket{s} = \ket{s+1},\; Z\ket{s} = \omega^s \ket{s}.
\]
This gives the (linear) operators:
\[
X = \sum_{j = 0}^{d-1} \ket{j+1}\bra{j},\; Z = \sum_{j = 0}^{d-1} \omega^j \ket{j}\bra{j}.
\]
From this, the identities $X^d = I$ and $Z^d = I$ are obvious. The generalized Hadamard gate is defined as~\cite{wang2020qudits}:
\begin{equation}
\label{eq:hadamard}
H \ket{s} = \frac{1}{\sqrt{d}} \sum_{j = 0}^{d-1} \omega^{js} \ket{j}    
\end{equation}

\begin{proposition} 
\label{prop:hadamard-unitary}
The generalized $X$, $Z$ and the Hadamard operators $H$ are unitary.
\end{proposition}
\begin{proof}
The fact that $X$ and $Z$ are unitary can be easily checked through their definitions. For $H$, we see that
\begin{align*}
    HH^\dagger \ket{s} &= H\left( \frac{1}{\sqrt{d}} \sum_{j = 0}^{d-1} \omega^{-js} \ket{j}\right)\\
            &= \frac{1}{\sqrt{d}} \sum_{j = 0}^{d-1} \omega^{-js} H\ket{j}\\
            &= \frac{1}{d} \sum_{j = 0}^{d-1} \omega^{-js} \left(\sum_{k = 0}^{d-1} \omega^{kj} \ket{k}\right)\\
             &= \frac{1}{d} \sum_{k = 0}^{d-1} \left(\sum_{j = 0}^{d-1} (\omega^{j})^{k-s}\right)  \ket{k}
\end{align*}
According to Proposition~\ref{prop:unity}, the term within the bracket is 0, unless $k - s \equiv 0 \pmod{d} \Rightarrow k \equiv s \pmod{d}$, in which case the sum is equal to $d$. Therefore, we get:
\begin{align*}
    HH^\dagger \ket{s} = \frac{1}{d} \cdot d \ket{s} = \ket{s}
\end{align*}
Similarly, $H^\dagger H = I$.
\end{proof}

\begin{proposition} 
\label{prop:ZX-interchange}
We have $HXH^\dagger = Z$ and $H^\dagger{Z}H = X$. 
\end{proposition}
\begin{proof}
\begin{align*}
    HXH^\dagger \ket{s} &= HX \left( \frac{1}{\sqrt{d}} \sum_{j = 0}^{d-1} \omega^{-js} \ket{j}\right)\\
     &= H\left( \frac{1}{\sqrt{d}} \sum_{j = 0}^{d-1} \omega^{-js} \ket{j+1}\right)\\
     &= \frac{1}{d} \sum_{j = 0}^{d-1} \omega^{-js} \left(\sum_{k = 0}^{d-1} \omega^{k(j+1)} \ket{k}\right)\\
      &= \frac{1}{d} \sum_{k = 0}^{d-1} \omega^k
      \left(\sum_{j = 0}^{d-1} (\omega^{j})^{k-s}\right)  \ket{k} \\
      &= \frac{1}{d} \omega^s \cdot d \ket{s} \\
      &= \omega^s \ket{s} = Z\ket{s}
\end{align*}
where again we have used Proposition~\ref{prop:unity}. Now applying Proposition~\ref{prop:hadamard-unitary} to $HXH^\dagger = Z$, we obtain $H^\dagger Z H = X$.
\end{proof}

Finally, we have:

\begin{proposition}
\label{pro:ex-es}
Let $s$ be an integer. Then $(X^s)^\dagger = X^{-s}$ and $(Z^s)^\dagger = Z^{-s}$.
\end{proposition}
\begin{proof}
Since $X$ is unitary, we have $X^{\dagger}X = I = X^\dagger X$. Therefore $(X^s)^\dagger X^s = I$. Now, $X^s \ket{r} = \ket{r + s}$. Therefore $(X^s)^\dagger X^s \ket{r} = (X^s)^\dagger \ket{r+s} \Rightarrow (X^s)^\dagger \ket{r+s} = \ket{r}$. Therefore, $(X^s)^\dagger = X^{-s}$. The proof for $Z$ is similar.
\end{proof}

\descr{Measurements.} To observe the current state of a quantum system, we need to measure the system. We will be using \emph{projective} measurements which are described by an \emph{obvervable}, i.e., a normal operator $M$ acting on the state space. According to the spectral theorem~\cite[\S 2.1.5]{nielsen2002quantum}, any normal operator is \emph{diagonalizable}, i.e., it can be written in terms of its eigenvalues and eigenvectors. Thus, we can write $M$ as
\begin{equation}
\label{eq:measurement}
    M = \sum_{m = 1}^d \lambda_m \ket{m}\bra{m},
\end{equation}
where $\lambda_m$ are the eigenvalues of $M$, $\ket{m}$ the corresponding eigenvectors, and $\{\ket{m}\}$ is an orthonormal basis of the state space. The measurement outcomes are precisely the eigenvalues, which we can map to integer outcomes: $\lambda_i \mapsto i$. The probability of obtaining the outcome $\lambda_m$ when measuring the state $\ket{\psi}$ is given as 
\begin{equation}
\label{eq:p_m}
p(m) = \bra{\psi}( \ket{m}\bra{m} )\ket{\psi} = \braket{\psi}{m}\braket{m}{\psi} =  \braket{\psi}{m}  \braket{\psi}{m}^* = | \braket{\psi}{m}|^2,    
\end{equation}
and the state of the system after the measurement changes to
\[
\frac{( \ket{m}\bra{m} )\ket{\psi}}{\sqrt{p(m)}} = \frac{\braket{m}{\psi}\ket{m}}{\sqrt{p(m)}} = \frac{\braket{m}{\psi}\ket{m}}{| \braket{\psi}{m}|} = \frac{\braket{m}{\psi}}{| \braket{m}{\psi}|} \ket{m},
\]
which is effectively $\ket{m}$ up to a global phase factor, i.e., $\braket{m}{\psi}/| \braket{m}{\psi}|$ of modulus 1, which can be ignored. 
We will be doing measurements in either the $Z$ or $X$ basis, i.e., either the observable $Z$ or $X$. Here, for instance, if measuring in the $Z$ (computational basis), the outcome is one of the $d$ eigenvalues $\omega^j$, where $0 \le j \le d - 1$, since the spectral decomposition of $Z$ is
\begin{equation}
\label{eq:z-measure}
Z = \sum_{j = 0}^{d-1} \omega^j \ket{j}\bra{j}.    
\end{equation}
A consequence of Proposition~\ref{prop:ZX-interchange} is that the $Z$ and $X$-basis are interchangeable via a Hadamard transform. For instance, to measure in the $X$ basis, whose eigenvectors are~\cite{wootton2011anyon}:
\[
\ket{\lambda} = \frac{1}{\sqrt{d}} \sum_{j = 0}^{d-1} \omega^{j\lambda} \ket{j}
\]
with eigenvalues $\omega^{-\lambda}$, we can apply the Hadamard operator $H$ and then measure in the computational basis $Z$ whose eigenvectors are $\ket{j}$ with corresponding eigenvalues $\omega^j$, as seen above. 

\descr{Multiple Qudits and Operations.} We shall be working on a composite system, i.e., a system with more than one qudit. The state space of a composite system is described as the tensor product of the state space of the individual qudits. If the $i$th qudit is denoted $\ket{\psi_i}$ then the state space of the $n$-qudit system is denoted $\ket{\psi_1} \otimes \cdots \otimes \ket{\psi_n}$. If $A$ is an operator on $\ket{\psi_1}$ and $B$ is an operator on $\ket{\psi_2}$, then the combined operation on the tensor product of these two states is defined as 
\begin{equation}
\label{eq:tensor-linear}
    (A \otimes B)(\ket{\psi_1} \otimes \ket{\psi_2}) = A\ket{\psi_1} \otimes B \ket{\psi_2}
\end{equation}
$A \otimes B$ is also a linear operator, and the definition easily extends to higher order systems. A couple of other useful properties of tensor products are as follows~\cite[\S 2.1.7]{nielsen2002quantum}:
\begin{align}
\label{eq:tensor-distrib}
    (\ket{\psi_1} \otimes \ket{\psi_2})(\bra{\psi_1} \otimes \bra{\psi_2}) &= \ket{\psi_1}\bra{\psi_1} \otimes \ket{\psi_2}\bra{\psi_2},\\
    \text{and } z(\ket{\psi_1} \otimes \ket{\psi_2}) &= (z\ket{\psi_1}) \otimes \ket{\psi_2} = \ket{\psi_1} \otimes (z\ket{\psi_2})\nonumber,
\end{align}
for a complex number $z$. Finally the tensor product distributes over addition of operators and states. To avoid excessive notation, we may write $\ket{\psi_1} \otimes \ket{\psi_2}$ as $\ket{\psi_1}\ket{\psi_2}$ or even $\ket{\psi_1\psi_2}$, as is conventional. If all qudits in the $n$ systems are in the state $\ket{\psi}$, then we use the compact notation $\ket{\psi}^{\otimes n}$ for the tensor product. For succinct representation, if $\ket{\mathbf{s}} = \ket{s_1} \otimes \cdots \otimes \ket{s_n}$, then we may also write it as $(s_1, s_2, \ldots, s_n)$, where each $s_i$ denotes the $i$th ket vector in the tensor product. The Hadamard operation on multiple qudits is defined as follows. Let $\mathbf{s} = (s_1, s_2, \ldots, s_n)$ and $\mathbf{j} = (j_1, j_2, \ldots, j_n)$, where $s_i, j_i \in \mathbb{Z}_d$. Then,
\begin{align*}
    H^{\otimes n}\ket{\mathbf{s}} &= \left(\frac{1}{\sqrt{d}} \sum_{j_1 = 0}^{d-1} \omega^{j_1s_1} \ket{j_1}\right) \cdots   \left(\frac{1}{\sqrt{d}} \sum_{j_n = 0}^{d-1} \omega^{j_ns_n} \ket{j_n}\right) \\
    &= \frac{1}{\sqrt{d^{n}}} \sum_{\mathbf{j} \in \mathbb{Z}^n_d} \omega^{\langle \mathbf{j}, \mathbf{s} \rangle} \ket{\mathbf{j}}
\end{align*}
In particular, if each $s_i = s$ are the same, we see that $\langle \mathbf{j}, \mathbf{s} \rangle = s \sum_{i=1}^n j_i = s \lVert \mathbf{j} \rVert$, where $\lVert \cdot \rVert$ is the $\ell_1$-norm. Thus,
\begin{equation}
\label{eq:hadamard-same-s}
    H^{\otimes n}\ket{s}^{\otimes n} = \frac{1}{\sqrt{d^{n}}} \sum_{\mathbf{j} \in \mathbb{Z}^n_d} \omega^{s \lVert \mathbf{j} \rVert} \ket{\mathbf{j}}.
\end{equation}

The (generalized) controlled-$X$ operator, denoted $CX$, is defined as the operator whose action is:
\begin{equation}
\label{eq:cnot}
    \ket{s}\ket{r} \rightarrow \ket{s}\ket{r + s} = \ket{s}X^s \ket{r}.
\end{equation}
Its conjugate is denoted as $CX^{\dagger}$, defined as the operator which maps $\ket{s}\ket{r}$ to $\ket{s}X^{-s} \ket{r} =\ket{s}\ket{r - s}$. We can see that the $CX$ gate is unitary as well due to Proposition~\ref{pro:ex-es}. 

\descr{Density Operator and Partial Trace.} Sometimes it is convenient to write the state of a system using the \emph{density matrix}. For an $n$-qudit system in state $\ket{\psi}$, its density matrix is defined as:
\[
\rho = \ket{\psi}\bra{\psi}
\]
An application of an operator $A$ on the qudit $\ket{\psi}$ in the density matrix representation is given as: $A\rho A^\dagger$, which is again another density matrix. An advantage of the density matrix representation of the state of the system is that we can find the description of subsystems in the composite state space. In particular if $\rho$ describes a system of two qudits, then we can find the state of the system of the first qudit as:
\[
\rho^1 = \text{tr}_2(\rho)
\]
where $\text{tr}_2()$ is an operator, called partial trace, defined as
\[
\text{tr}_2(\ket{a_1}\bra{b_1} \otimes (\ket{a_2}\bra{b_2})) = \ket{a_1}\bra{b_1} \text{tr}(\ket{a_2}\bra{b_2})
\]
where $\ket{a_i}$ and $\ket{b_i}$ are any two vectors in the state space $i$. Here, $\text{tr}()$ is the matrix trace, which for the cross product is defined as: $\text{tr}(\ket{a_2}\bra{b_2}) = \braket{a_2}{b_2}$. The definition of the partial trace is made complete by further noting that it is linear in its input. For instance, assume $n = 2$, and client $i$ owns qudit $\ket{\psi_i}$. Suppose the system is in the state $\ket{\psi} = \ket{\psi_1}  \otimes \ket{\psi_2}$. Then the density operator representation of the system is:
\[
\rho = \ket{\psi}\bra{\psi} = \left( \ket{\psi_1} \otimes \ket{\psi_2} \right) \left( \bra{\psi_1} \otimes \bra{\psi_2}\right) = \ket{\psi_1}\bra{\psi_1} \otimes \ket{\psi_2}\bra{\psi_2}
\]
Then the state of system 1, i.e., client 1, can be obtained via the partial trace:
\[
\rho^1 = \text{tr}_2(\rho) = \ket{\psi_1}\bra{\psi_1} \text{tr}(\ket{\psi_2}\bra{\psi_2}) = \ket{\psi_1}\bra{\psi_1} \braket{\psi_2}{\psi_2} = \ket{\psi_1}\bra{\psi_1},
\]
since the qudit $\ket{\psi_2}$ is defined as a unit vector. This is the state of system 1 as we would expect. The usefulness of the partial trace is more prominent when we consider entangled states, i.e., states that cannot be written as tensor products of the constituent states. 

We can also describe measurements in the density operator representation. Let $\ket{m}\bra{m}$ denote a measurement operator with outcome $m$ as in Eq.~\eqref{eq:measurement}. Then, given the state $\rho$, the probability of outcome $m$ is given by $p(m) = \text{tr}(\ket{m}\bra{m} \rho)$, and the post-measurement state is:
\[
\rho' = \frac{\ket{m}\bra{m} \rho \ket{m}\bra{m}}{p(m)}.
\]

% \descr{Quantum Differential Privacy}. Just like the classical case, we call a quantum channel $\mathcal{N}$ $(\epsilon, \delta)$ differentially private for $\epsilon, \delta \geq 0$ if for any input quantum states $\rho$ and $\sigma$ such that $\tau(\rho, \delta)\leq \kappa$ for some $\kappa\in (0,1]$, satisfies
% \[
% Pr[\mathcal{N}(\rho)\in S]\leq e^{\epsilon} Pr[\mathcal{N}(\sigma)\in S] + \delta
% \]

% where $S\subset Out(M)$ for any POVM $M$ and $\tau$ is the trace distance of the density matrices. 
\section{Proposed Protocol}
\label{sec:protocol}
Our quantum protocol for Algorithm~\ref{algo:bounded-noise} in the shuffle model is based on the e-voting protocol from~\cite{hillery2006voting} as well as the protocol from~\cite{qa-transmissions} for anonymously broadcasting a single bit within a group of nodes, with major differences which we explain in Section~\ref{sec:rw}. The protocol is described in Protocol~\ref{protocol:shuffle}. Before the start of the protocol, we assume that each client and the server share at least one generalized Bell state (Eq.~\eqref{eq:bell}). Here sharing means that the client owns one qudit and the server the other. The Bell state is written as a linear combination of states $\ket{jj} = \ket{j} \otimes \ket{j}$, for $0 \leq j \leq d-1$. Here, we assume that the first qudit in the tensor product is owned by the server, and the second by the client. Note also that we cannot write this state as a tensor product of two separate states $\ket{k} \otimes \ket{k}$. This means that these states are \emph{entangled}. 

This Bell state is used to create a quantum channel to \emph{teleport} the state $\ket{\psi_0}$ in Eq.~\ref{eq:ghz} which is called a GHZ state. This is again an entangled state. The server sends each qudit in this state to a separate client. To reliably send the GHZ state, one needs classical communication between the client and the server. This is shown in Figure~\ref{fig:teleport}, and we shall discuss it in detail in the next section. Client $i$ then applies the $Z$ operator to the qudit of the GHZ state sent by the server a total of $y_i$ times, where $y_i$ is the output from the LDP Algorithm~\ref{algo:bounded-noise}. The clients then apply the Hadamard gate to their qudits before measuring the result in the computational basis. After the measurement by each client we need a classical channel between the client and the server to send the measurement outcome to the server. The classical and quantum channels complete the picture of our communication network as shown in Figure~\ref{fig:channels}. In this section we assume that the processing of quantum circuits as well as communication of quantum states happens error free. In the next section, we provide details of these circuits, followed by their fault-tolerant implementation. For now, we show that the protocol is correct and secure against honest-but-curious clients and the server. 

\begin{figure}
\centering
%\resizebox{\columnwidth}{!}{
\begin{tikzpicture}
%Nodes
\node[circle, minimum width = 2cm, minimum height = 2cm, draw=black, fill=yellow!10, very thick] (client) {Client};

\node[circle, minimum width = 2cm, minimum height = 2cm, right=6cm of client.east, draw=black, fill=gray!10, very thick] (server) {Server};

%Channels

\node[] 							(bcstart) 			[yshift=0.3cm, left=0.08cm of client.east] {};

\node[] 							(bcend) 			[yshift=0.3cm, right=0.08 cm of server.west] {};

\draw[very thick] (bcstart) -- node [above] {Quantum channel via Bell pairs $(\leftarrow)$} (bcend.west);

\node[] 							(ccstart1) 			[yshift=-0.3cm, left=0.08cm of client.east] {};

\node[] 							(ccend1) 			[yshift=-0.3cm, right=0.08 cm of server.west] {};

\draw[] (ccstart1.east) -- (ccend1.west);

\node[] 							(ccstart2) 			[yshift=-0.35cm, left=0.08cm of client.east] {};

\node[] 							(ccend2) 			[yshift=-0.35cm, right=0.08 cm of server.west] {};

\draw[] (ccstart2.east) -- node [below] {Classical channel $(\leftrightarrows)$}(ccend2.west);
\end{tikzpicture}
%}
\caption{The two types of channels between each client and the server. The arrows indicate the direction in which information flows in our protocol using the respective channel.}
\label{fig:channels}
\end{figure}
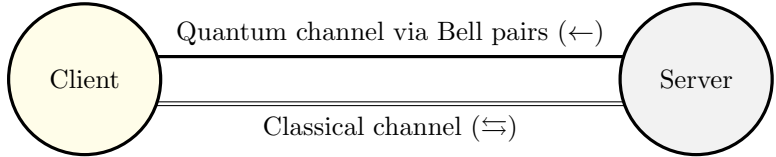

\RestyleAlgo{boxed}
\begin{algorithm}[!ht]
\SetAlgorithmName{Protocol}{List of protocols}
\SetAlgoLined
\LinesNumbered
%
%\SetAlCapSkip{1em}
\DontPrintSemicolon{}
%\let\oldnl\nl% Store \nl in \oldnl
%\newcommand{\nonl}{\renewcommand{\nl}{\let\nl\oldnl}}% Remove line number for one line
%\nonl\TitleOfAlgo{{Select Features}}
\SetKwInOut{Input}{Input}
\SetKwInOut{State}{State}
\Input{Number of clients $n$ and parameter $\kappa$ of Algorithm~\ref{algo:bounded-noise}, the server chooses $d > (\kappa-1)n$. The server shares a generalized Bell state with each client
\begin{align}
%\begin{equation}
\label{eq:bell}
  \ket{\beta} &= \frac{1}{\sqrt{d}} \sum_{j = 0}^{d-1} \ket{jj}  
%\end{equation}
\end{align}
\hrulefill
}
%\State{A noise dictionary, denoted $\mathsf{nd}$, with keys from subsets of $U$ and values in $\mathbb{Z}_{\pm r}$.}
The server prepares the initial generalized GHZ state:
\begin{equation}
\label{eq:ghz}
\ket{\psi_0} = \frac{1}{\sqrt{d}} \sum_{j = 0}^{d-1} \ket{j}^{\otimes n}
\end{equation}\;
For $i = 1$ to $n$, the server teleports qudit $i$ in $\ket{\psi_0} $ using the generalized Bell state $\ket{\beta}$ to client $i$ through the circuit shown in Figure~\ref{fig:teleport}.\;
Client $i$ runs LDP Algorithm~\ref{algo:bounded-noise} with input $x_i$ to obtain $y_i$.\;
Client $i$ applies $Z^{y_i}$ to its qudit, where
\begin{equation}
\label{eq:unit-yes}
Z = \sum_{j=0}^{d-1} \omega^j \ket{j}\bra{j}.
\end{equation}\;
Client $i$ applies the Hadamard gate $H$ (Eq.~\ref{eq:hadamard}) to his/her qudit.\;
Client $i$ measures his/her qudit in the computational basis $Z$, with possible outcomes $z_i \in \{0, 1, \ldots, d- 1\}$ corresponding to eigenvalues $\{\omega^0, \omega, \ldots, \omega^{d-1}\}$.\;
All clients send their measurement outcomes $z_i$ to the server.\;
The server computes $z \equiv - \sum_{i=0}^n z_i \pmod{d}$, and outputs the de-biased sum via Eq.~\eqref{eq:de-bias} as:
\[
\frac{1}{1 - \gamma}\left( z - \frac{\gamma (\kappa-1)n}{2} \right).
\]\;
\caption{Our proposed protocol for Algorithm~\ref{algo:bounded-noise} in the shuffle model.}
\label{protocol:shuffle}
\end{algorithm}

\begin{theorem}
\label{theorem:correctness}
If all clients execute the protocol honestly, then the server obtains $z$ (Step 8), which is the sum of differentially private outputs from all clients.
\end{theorem}
\begin{proof}
First note that each client's differentially private output $y_i$ is from the set $\{0, 1, \ldots, \kappa - 1\}$. Thus, $0 \leq \sum_{i = 1}^n y_i \leq (\kappa-1)n < d$. Assume the sum of the outputs is $m = \sum_{i = 1}^n y_i$. We shall show that the server exactly extracts $m$ from the protocol. After all $n$ clients have applied $Z^{y_i}$ (where $Z$ is given by Eq.~\ref{eq:unit-yes}) to their qudits, the entangled state $\ket{\psi_0}$ becomes
\begin{equation}
    \ket{\psi_m} = \frac{1}{\sqrt{d}}\sum_{j=0}^{d-1} (\omega^{j})^m \ket{j}^{\otimes n}.
\end{equation}

Now, each player applies the hadamard operator to his/her qubit, resulting in the state:
\begin{align}
    H^{\otimes n} \ket{\psi_m} &= \frac{1}{\sqrt{d}} \sum_j (\omega^j)^m H^{\otimes n} \ket{j}^{\otimes n} \nonumber\\
    &= \frac{1}{\sqrt{d}} \sum_j (\omega^j)^m \left( \frac{1}{\sqrt{d^n}} \sum_{\mathbf{k} \in \mathbb{Z}^n_d} \omega^{\langle \mathbf{k}, \mathbf{j} \rangle} \ket{\mathbf{k}} \right)\nonumber \\
    &=\frac{1}{\sqrt{d}} \sum_j (\omega^j)^m \left( \frac{1}{\sqrt{d^n}} \sum_{\mathbf{k} \in \mathbb{Z}^n_d} \omega^{j \lVert \mathbf{k} \rVert} \ket{\mathbf{k}} \right)\nonumber \\
    &=\frac{1}{\sqrt{d^{n+1}}} \sum_{\mathbf{k} \in \mathbb{Z}^n_d} \sum_j (\omega^j)^{m + \lVert \mathbf{k} \rVert} \ket{\mathbf{k}} \nonumber \\
    &=\frac{d}{\sqrt{d^{n+1}}} \sum_{\substack{\mathbf{k} \in \mathbb{Z}^n_d \\ m + \lVert \mathbf{k} \rVert \equiv 0}} \ket{\mathbf{k}} = \frac{1}{\sqrt{d^{n-1}}} \sum_{\substack{\mathbf{k} \in \mathbb{Z}^n_d \\ m + \lVert \mathbf{k} \rVert \equiv 0}} \ket{\mathbf{k}} \label{eq:post-hadamard-state},
\end{align}
where we have used Proposition~\ref{prop:unity} and Eq.~\eqref{eq:hadamard-same-s}. How many vectors $\mathbf{k} \in \mathbb{Z}^n_d$ satisfy $m + \lVert \mathbf{k} \rVert \equiv 0 \pmod{d}$ for a given $m$? One way to count this is to note that we can choose any number in $\mathbb{Z}_d$ for the first $n-1$ entries in $\mathbf{k}$, giving us $d^{n-1}$ possibilities. However, the last entry should be fixed as $k_n \equiv -m -\sum_{i=1}^{n-1}k_i \pmod{d}$. Thus, we have a total of $d^{n-1}$ such vectors. Furthermore, since $m < d$, each value of $m$ gives a different $k_n$. Therefore, each value of $m$ gives a subset of the space $\mathbb{Z}_d^n$, having $d^{n-1}$ vectors each, and furthermore all these subsets are disjoint, giving us a total of $(n+1)d^{n-1}$ vectors. Note that this division is a partition only if $d = n+1$.\footnote{Together with the condition $d > (\kappa - 1)n$ and noting that $n$ needs to be at least 2, we see that this is possible only if $\kappa = 2$, i.e., binary randomized response.}  

Next, for the measurement of the state in Eq.~\eqref{eq:post-hadamard-state}, consider the vector $\ket{\mathbf{z}} \in \mathbb{Z}_d^n$, whose $i$th entry, i.e., $\ket{z_i}$, corresponds to the $i$th measurement outcome of client $i$ (see Eq.~\eqref{eq:z-measure}). If $\ket{\mathbf{z}}$ is such that $m + \lVert \mathbf{z} \rVert \not\equiv 0 \pmod{d}$, then for all vectors $\ket{\mathbf{k}}$ in the sum of Eq.~\eqref{eq:post-hadamard-state}, there exists at least one $\ket{z_i}$, where $1 \le i \le n$, such that $z_i \neq k_i$, and hence 
\begin{align*}
\braket{\mathbf{z}}{\mathbf{k}} &= \braket{z_1}{k_1} \otimes \cdots \otimes \braket{z_i}{k_i} \otimes \cdots \otimes \braket{z_n}{k_n}    \\
        &= \braket{z_1}{k_1} \cdots 0 \cdots \braket{z_n}{k_n} = 0
\end{align*}
Hence, only those outcomes $ \mathbf{z}$ which satisfy $m + \lVert \mathbf{z} \rVert \equiv 0 \pmod{d}$ are probable due to Eq~\eqref{eq:p_m}. For any such outcome, the server can compute:
\[
m \equiv - \lVert \mathbf{z} \rVert \pmod{d},
\]
as the unique outcome.  
\end{proof}

Next we prove the (information-theoretic) privacy of the system. 

\begin{theorem}
\label{theorem:privacy}
Any honest-but-curious client does not learn the differentially private input $y_i$ of any other client in Protocol~\ref{protocol:shuffle}, except what is discernible through $\sum_i^{m} y_i$. Furthermore, the server does not learn the individual differentially private inputs $y_i$ of all clients, except what is discernible through $\sum_i^{m} y_i$. 
\end{theorem}
\begin{proof}
Let $m$ denote the sum of the differentially private outputs fromm all clients. Then the state becomes:
\begin{equation}
    \ket{\psi_m} = \frac{1}{\sqrt{d}}\sum_{j=0}^{d-1} (\omega^{j})^m \ket{j}^{\otimes n}.
\end{equation}
On their own, the clients cannot determine the phase. To see this, the density operator for the above state is:
\[
\rho_m = \ket{\psi_m} \bra{\psi_m} = \frac{1}{d} \sum_j \sum_k (\omega^{j-k})^m \ket{j}^{\otimes n}\bra{k}^{\otimes n}.
\] 
For any client, tracing out the rest of the system reveals:
\begin{align*}
    \frac{1}{d}\sum_{j}\sum_{k}(\omega^{j-k})^m \ket{j}\bra{k} \delta_{j,k}
    &= \frac{1}{d}\sum_{j}(\omega^{j-j})^m \ket{j}\bra{j} \\
    &= \frac{1}{d}\sum_{j} \ket{j}\bra{j},
\end{align*}
which is the same state regardless of the value of $m$, and hence no information is leaked. Next, let us examine the state of a client after the application of the Hadamard gate (Step 5). The density matrix of the system is given as:
\[
\rho'_m = \frac{1}{d^{n-1}} \sum_{\mathbf{k}} \sum_{\mathbf{j}} \ket{\mathbf{k}} \bra{\mathbf{j}}, 
\]
where the sum is understood to be over all vectors $\mathbf{k},  \mathbf{j} \in \mathbb{Z}^n_d$ satisfying $m + \lVert \mathbf{k} \rVert \equiv m + \lVert \mathbf{j} \rVert \equiv 0 \pmod{d}$. Without loss of generality, assume we measure the state of player 1. Tracing out the rest of the system (applying partial trace) yields:
\begin{align*}
   \frac{1}{d^{n-1}} \sum_{\mathbf{k}} \sum_{\mathbf{j}} \ket{k_1}\bra{j_1} \braket{k_2}{j_2} \cdots \braket{k_n}{j_n}
\end{align*}
All dot products $\braket{k_i}{j_i}$ with $2 \le i \le n$ are 1 only if $k_i = j_i$. But since $m + \sum_{i = 1}^{n} k_i \equiv m + \sum_{i = 1}^{n} j_i \pmod{d}$, this means that $k_1 = j_1$ as well. Therefore, we get
\[
 \frac{1}{d^{n-1}} \sum_{\mathbf{k}}  \ket{k_1}\bra{k_1} 
\]
As noted in the proof of Theorem~\ref{theorem:correctness}, for each choice of $k_1 \in \mathbb{Z}_d$, there are $d^{n-2}$ possible vectors $\mathbf{k}$ satisfying $m + \lVert \mathbf{k} \rVert \equiv 0 \pmod{d}$. Furthermore, this count remains the same if we replace player 1 with any other player. Hence each player has the same state, again regardless of the value of $m$. 

Finally we look at the measurement outcomes of a player. Again, without loss of generality, assume it is player 1 who measures first. Let us fix a measurement result $\ell$. We get:
\begin{align*}
    p(\ell) &= \frac{1}{d^{n-1}} \sum_{\mathbf{k}} \sum_{\mathbf{j}} \text{tr}\left((\ket{\ell}\bra{\ell} \otimes I_2 \otimes \cdots \otimes I_n)(\ket{\mathbf{k}} \bra{\mathbf{j}})\right)\\
    &=\frac{1}{d^{n-1}} \sum_{\mathbf{k}} \sum_{\mathbf{j}} \text{tr}\left( \ket{\ell}\braket{\ell}{k_1} \bra{j_1} \otimes \ket{k_2}\bra{j_2} \otimes \cdots \otimes \ket{k_n}\bra{j_n} \right)\\
    &=\frac{1}{d^{n-1}} \sum_{\mathbf{k}} \sum_{\mathbf{j}} \braket{\ell}{k_1} \braket{\ell}{j_1} \braket{k_2}{j_2} \cdots \braket{k_n}{j_n}\\
    &=\frac{1}{d^{n-1}} \sum_{\mathbf{k}}  \braket{\ell}{k_1} \braket{\ell}{k_1} = \frac{1}{d^{n-1}} (d^{n-2}) = \frac{1}{d},
\end{align*}
where again there are exactly $d^{n-2}$ possible vectors $\mathbf{k}$ satisfying $m + \lVert \mathbf{k} \rVert \equiv 0 \pmod{d}$ with $k_1 = \ell$. Hence each of the $d$ possible outcomes are equally likely for each player. 

The post-measurement state is:
\[
\frac{\text{tr}(\ket{\ell}\bra{\ell} \rho'_m \ket{\ell}\bra{\ell}}{p(\ell)} = \frac{1}{d^{n-2}} \sum_{\substack{\mathbf{k} \in \mathbb{Z}^n_d \\ m + \lVert \mathbf{k} \rVert \equiv 0\\ k_1 = \ell}} \sum_{\substack{\mathbf{j} \in \mathbb{Z}^n_d \\ m + \lVert \mathbf{j} \rVert \equiv 0\\ j_1 = \ell}} \ket{\mathbf{k}}\bra{\mathbf{j}}.
\]
Thus, effectively, we are in a system with one client less. Therefore, the analysis for the state of the current system and post-measurement is the same as before with $n$ replaced with $n-1$.

For privacy from the server, note that since all $d$ possible outcomes are equally likely for each client, the received measurement outcome $z_i$ from client $i$ is independent of its differentially private input $y_i$. Hence, the server does not learn the differentially private input $y_i$.  
\end{proof}

\descr{Efficiency.} As we discuss in the Sections~\ref{sec:circuits} and \ref{sec:fault-tolerance} our protocol can be implemented using quantum circuits involving only \emph{Clifford} gates. See Section~\ref{sec:fault-tolerance} for a definition. A key property of the Clifford group of gates is that they can be efficiently simulated by a classical computer; see, for example~\cite[\S 6.2]{gottesman2016errorbook}. The reason for this is that the Clifford group on $n$ qudits only consists of up to $3dn$ individual qudit gates of the type $X^{j}$ and $Z^k$ together with a phase (see Eq.~\eqref{eq:pauli-group}), whereas a general unitary gate on $n$ qudits needs up to $d^{2n}$ parameters to be specified. Thus our protocol is highly efficient.  
% We work in the generalized Pauli group of qudits~\cite{havlicek2007gen-pauli}. Let $\mathbb{C}^d$ be the $d$-dimensional complex Hilbert space whose basis elements form the set $\{\ket{s}\}$ for all $s \in \mathbb{Z}_d$. 

% For our protocol, the initial state to be shared with $n$ participants is prepared as:
% \begin{equation}
% \label{eq:ghz}
% \ket{\psi_0} = \frac{1}{\sqrt{d}} \sum_{j = 0}^{d-1} \ket{j}^{\otimes n},
% \end{equation}
% where $n$ is the number of players, and $d > n$.  This is a generalized GHZ state for qudits. Each player has one qudit. Any player sending a bit, applies the unitary transformation:
% \begin{equation}
% \label{eq:unit-yes}
% Z = \sum_{j=0}^{d-1} \omega^j \ket{j}\bra{j}.
% \end{equation}

% Thus, the value of $m$ partitions the vector space $\mathbb{Z}_d^n$ into $d$ subspaces each containing $d^{n-1}$ vectors, i.e., those vectors satisfying $m + \lVert \mathbf{k} \rVert \equiv 0 \pmod{d}$. Each players measures in the $Z$-basis (computational basis), which means each of these $d^{n-1}$ vectors are equally likely to be the outcome. Each player sends his/her qudit value using the classical communication channel to the server. Given any value of $m$ each of the $d$ qudits are equally likely to be the outcome of an individual measurement. But since the vector as a whole satisfies $m + \lVert \mathbf{k} \rVert \equiv 0 \pmod{d}$, summing up the individual qudits reveals $-m$ modulo $d$, from which the server can recover $m$ in a straightforward manner. 

\section{Quantum Circuits}
\label{sec:circuits}
In this section, we show the circuits constructing the different quantum states in the description of Protocol~\ref{protocol:shuffle}, and show that the construction is correct.  

\descr{Generalized Bell State.} The circuit to create the generalized Bell state from Eq.~\ref{eq:bell} is shown in Figure~\ref{fig:bell-state}. This is analogous to the construction of the Bell state $\frac{1}{\sqrt{2}}( \ket{00} + \ket{11})$ for qubit systems~\cite[\S 1.3.6]{nielsen2002quantum}. After the application of the Hadamard gate on the first qudit, according to Eq.~\eqref{eq:hadamard}, the state is transformed to:
\[
\ket{0} \otimes \ket{0} \rightarrow \frac{1}{\sqrt{d}}\sum_{j=0}^{d-1} \ket{j} \otimes \ket{0}
\]
Next, the $CX$ gate is applied which according to Eq.~\eqref{eq:cnot} transforms the state to:
\[
\frac{1}{\sqrt{d}}\sum_{j=0}^{d-1} \ket{j} \otimes \ket{0} \rightarrow \frac{1}{\sqrt{d}}\sum_{j=0}^{d-1} \ket{jj},
\]
as required. Note that the `dot' in the figure represents the control qudit for the $CX$ gate. We assume that at the time of registration, the server shares multiple Bell states per client with cardinality equal to the number of times the LDP mechanism is invoked (see Proposition~\ref{prop:composition}). We assume that the first qudit of $\ket{\beta}$ is with the server and the second with the client. 

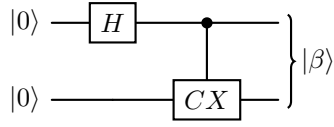
\begin{figure}
\begin{center}
\begin{quantikz}
\lstick{\ket{0}} & \gate{H} & \ctrl{1} & \rstick[2]{\ket{\beta}}\\ 
\lstick{\ket{0}} &  & \gate{CX} & 
\end{quantikz}
\end{center}
    \caption{Circuit for creating the generalized Bell state from Eq~\eqref{eq:bell}.}
    \label{fig:bell-state}
\end{figure}

\descr{The Initial GHZ State.} The circuit to prepare the initial GHZ state $\ket{\psi}_0$ from Eq.\eqref{eq:ghz} is shown in Figure~\ref{fig:ghz-state}. This is almost identical to the circuit for preparing the generalized Bell state, except now we have the initial state prepared as $\ket{0}^{\otimes n}$, and a total of $n$ $CX$ gates. 

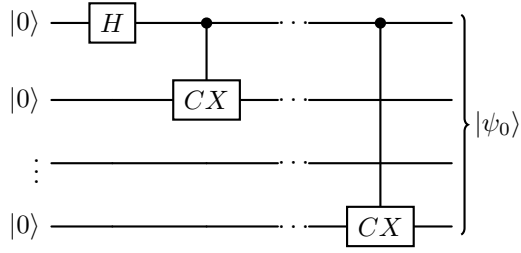
\begin{figure}
\begin{center}
\begin{quantikz}
\lstick{\ket{0}} & \gate{H} & \ctrl{1} & \ldots & \ctrl{3} & \rstick[4]{\ket{\psi_0}}\\ 
\lstick{\ket{0}} &  & \gate{CX} & \ldots & & \\
\lstick{\vdots} & & & \cdots & &\\
\lstick{\ket{0}} &  &  & \ldots & \gate{CX} & \\
\end{quantikz}
\end{center}
    \caption{Circuit for creating the initial GHZ state $\ket{\psi_0} = \frac{1}{\sqrt{d}} \sum_{j = 0}^{d-1} \ket{j}^{\otimes n}$ from Eq~\eqref{eq:ghz}.}
    \label{fig:ghz-state}
\end{figure}

\descr{Teleportation of the Initial GHZ State.} 
%The result can thus be communicated to the server in a fault tolerant way due to the use of the classical communication channel. What is left is to transmit the initial generalized GHZ state $\ket{\psi_0}$ to the $n$ parties. To do this reliably, we use the generalized Bell state:
% \[
% \ket{\beta} = \frac{1}{\sqrt{d}} \sum_{j = 0}^{d-1} \ket{jj}
% \]
% We assume that the first qudit of $\ket{\beta}$ is with the server and the second with the player. Thus, the server shares $n$ such Bell pairs with the $n$ players. Then 
Each of the $n$ qudits of the state $\ket{\psi_0}$ can be sent to a client via the circuit shown in Figure~\ref{fig:teleport} analogous to the quantum teleportation circuit for a qubit using the Bell state with $d = 2$~\cite[\S 1.3.7]{nielsen2002quantum}. The double wires indicate classical information. Let us assume a qudit $T$ of the initial state is getting teleported to a client. So the combined initial state with the Bell pair is:
\[
\ket{\psi_0}\ket{\beta} = \frac{1}{d} \sum_j \sum_k \ket{k}^{\otimes (n-1)}_{R} \ket{k}_{T}  \ket{jj}_{S C}  
\]

Where $\ket{j}_S$ is the server owned qudit of the Bell pair, $\ket{j}_C$ is the client owned qudit of the Bell pair, and $R$ denotes the rest of the initial state. Rearranging this we get:
\[
\frac{1}{d} \sum_j \sum_k \ket{kj}_{TS} \ket{k}^{\otimes (n-1)}_{R}   \ket{j}_{C}  
\]

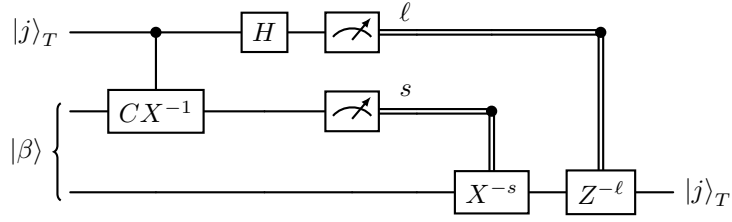
\begin{figure}
\begin{center}
\begin{quantikz}
\lstick{$\ket{j}_T$} & \ctrl{1} & \gate{H} & \meter[label style={right = 0.5cm}]{\ell} & \setwiretype{c} & & \ctrl{0}\\ 
\lstick[2]{\ket{\beta}} & \gate{CX^{-1}} &  & \meter[label style={right = 0.5cm}]{s} & \setwiretype{c} & \ctrl{0} \\
&&&&& \gate{X^{-s}}\wire[u][1]{c} & \gate{Z^{-\ell}} \wire[u][2]{c} & \rstick{$\ket{j}_T$}
\end{quantikz}
\end{center}
    \caption{Teleportation circuit for teleporting an individual qudit $\ket{j}_T$ of the GHZ state to a client using generalized Bell pairs $\ket{\beta}$. Here, the server has the top two qudits, and the client the bottom.}
    \label{fig:teleport}
\end{figure}

%Let the qudit state $\ket{\phi} = \sum_{j=0}^{d-1} \alpha_j\ket{j}$, where $\sum_{j=0}^{d-1} |\alpha_j|^2 = 1$. The initial state is:%

%After the $CX^{-1}$ gate, we get:
%\[
%\frac{1}{\sqrt{d}} \sum_j \sum_k \alpha_k \ket{k(j-k)j} 
%\]

After the $CX^{-1}$ gate, we get:\

\[
\frac{1}{d} \sum_j \sum_k \ket{k(j-k)}_{TS} \ket{k}^{\otimes (n-1)}_{R}   \ket{j}_{C} 
\]

Next is the Hadamard gate, giving us the state:
\[
\frac{1}{d\sqrt{d}} \sum_j \sum_k \sum_{\ell} \omega^{\ell k}\ket{\ell(j-k)}_{TS} \ket{k}^{\otimes (n-1)}_{R}   \ket{j}_{C} 
\]

%Next is the Hadamard gate, giving us the state:
%\[
%\frac{1}{d} \sum_j \sum_k \sum_\ell %\alpha_k \omega^{\ell k}\ket{\ell (j-k)j} 
%\]
Now fix an $\ell$ and $s \equiv j - k \pmod{d}$. Then, collecting terms we see that the term containing $\ket{\ell s}$ is:
%\[
%\frac{1}{d} \left( \ket{\ell s} \sum_j %\alpha_{j-s} \omega^{\ell(j-s)} \ket{j} + \cdots \right) 
%\]
\[
\frac{1}{d\sqrt{d}} \left( \ket{\ell s}_{T S} \sum_j \sum_k   \omega^{\ell k}\ket{k}^{\otimes (n-1)}_{R}   \ket{j}_{ C} + \ldots \right)
\]

Thus after measuring the servers qudits $T$ and $S$, if the outcome is $\ell$ and $s$, the whole state with the rest of the initial state and the qudit $C$ becomes:
\[
\frac{1}{\sqrt{d}} \sum_j \sum_{k} \omega^{\ell k} \ket{k}^{\otimes (n-1)}_{R}   \ket{j}_{C}
\]
The server conveys the result to the client via the classical channel after which the client can apply the correction operator $Z^{-\ell}X^{-s}$ and the state  becomes:

%Thus, if the server measures its qudits, and gets the outcome $\ket{\ell s}$, the player has the state:
%\[
%\sum_j \alpha_{j-s} \omega^{\ell(j-s)} \ket{j}.
%\]

%Once the server communicates $\ket{\ell s}$ to the player, through a classical channel, the client can then apply the operator $Z^{-\ell}X^{-s}$ to obtain:
%\begin{align*}
%   Z^{-\ell}X^{-s}  \sum_j \alpha_{j-s} \omega^{\ell(j-s)} \ket{j} &= \sum_j \alpha_{j-s} \omega^{\ell(j-s)} Z^{-\ell} X^{-s} \ket{j} \\
%  &= \sum_j \alpha_{j-s} \omega^{\ell(j-s)} Z^{-\ell} \ket{j-s}\\
%   &= \sum_j \alpha_{j-s} \omega^{\ell(j-s)} \omega^{-\ell(j-s)} \ket{j-s}\\
%   &= \sum_j \alpha_{j-s} \ket{j-s} = \ket{\phi}
%\end{align*}

\begin{align*}
   & Z^{-\ell}X^{-s} \frac{1}{\sqrt{d}} \sum_j \sum_{k} \omega^{\ell k} \ket{k}^{\otimes (n-1)}_{R}   \ket{j}_{C} \\
   &= \frac{1}{\sqrt{d}} \sum_j \sum_{k} \omega^{\ell k} \ket{k}^{\otimes (n-1)}_{R}  Z^{-\ell}X^{-s}  \ket{j}_{C} \\
  &= \frac{1}{\sqrt{d}} \sum_j \sum_{k} \omega^{\ell k} \ket{k}^{\otimes (n-1)}_{R}  Z^{-\ell} \ket{j-s}_C\\
  &= \frac{1}{\sqrt{d}}  \sum_{k} \omega^{\ell k} \ket{k}^{\otimes (n-1)}_{R}  Z^{-\ell} \ket{k}_C\\
   &= \frac{1}{\sqrt{d}}  \sum_{k} \omega^{\ell k} \ket{k}^{\otimes (n-1)}_{R} \omega^{-\ell k} \ket{k}_C\\
   &= \frac{1}{\sqrt{d}}  \sum_{k} \ket{k}^{\otimes (n-1)}_{R} \ket{k}_C\\
   &=\ket{\psi_0} 
\end{align*}

Notice that through this circuit one qudit gets teleported to the target client (with whom that particular Bell pair is shared) and an entanglement swapping happens between the original qudit of the initial state and the local qudit of the client. The server will perform this process for all other qudits in the GHZ state. Thus the whole initial state remains unchanged as we can see with the calculation above, while the qudits get distributed among clients. 

\descr{Local Operations for a Client.} The local operations done by each client on his/her qudit $\ket{\phi_i}$ are shown in Figure~\ref{fig:local-operations}. These are self-explanatory from the discussion thus far.

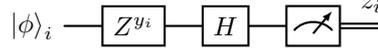
\begin{figure}
\begin{center}
\begin{quantikz}
\lstick{$\ket{\phi}_i$} & \gate{Z^{y_i}} & \gate{H} & \meter[label style={right = 0.5cm}]{z_i} & \setwiretype{c} 
\end{quantikz}
\end{center}
    \caption{Local operations for player $i$ having qudit $\ket{\phi}_i$ of the GHZ state $\ket{\psi_0}$ with differentially private output $y_i$.}
    \label{fig:local-operations}
\end{figure}

\subsection{Local Differential Privacy via Quantum Noise}
We have assumed that the differentially private algorithm is applied classically to each user's classical input. For completeness, we show that we can also implement the algorithm via a quantum circuit. Assume that the user's raw input is $x \in \{0, 1, \ldots, \kappa-1\}$. Note that $d > (\kappa-1)n$. We encode this classical input as a qudit $\ket{x}$. Let us calculate the probability that after the application of the LDP Algorithm~\ref{algo:bounded-noise}, the output of the player remains the same:
\begin{align*}
    \Pr[\ket{x} \rightarrow \ket{x}] &= \Pr[\ket{x} \rightarrow \ket{x} | b = 0] \Pr[b = 0] \\
    &+ \Pr[\ket{x} \rightarrow \ket{x} | b = 1] \Pr[b = 1] \\
    &= 1 \cdot (1 - \gamma) + \frac{1}{\kappa} \cdot \gamma = 1 - \frac{\kappa-1}{\kappa}\gamma.
\end{align*}
Let $p = 1 - (\kappa-1)\gamma/\kappa$. Let us also calculate the probability of the event $\ket{x} \rightarrow \ket{y}$, where $y \neq x$:
\begin{align*}
    \Pr[\ket{x} \rightarrow \ket{y}] &= \Pr[\ket{x} \rightarrow \ket{y} | b = 0] \Pr[b = 0] \\
    &+ \Pr[\ket{z} \rightarrow \ket{y} | b = 1] \Pr[b = 1] \\
    &= 0 \cdot (1 - \gamma) + \frac{1}{\kappa} \cdot \gamma = \frac{\gamma}{d} = \frac{1-p}{\kappa-1}.
\end{align*}

%quantum mechanica 
%Of course, we do not stop here. Just because the protocol is traceless, does not mean the server cannot learn the bit of a participant (e.g., if the server knows none of the other players sent a bit). We would like to use local quantum operations that would add noise into the bits sent by the players. The simplest case is when each player either sends a 0 or a 1. That is the domain is $\{0, 1\}$. In this case, a simple protocol samples a bit $b$ from a Bernoulli distribution with parameter $\gamma$. If $b = 0$, player answers truthfully. Otherwise, the player answers with a random bit~\cite[\S 3.1]{privacy-blanket-shuffle}. Suppose player $i$ has the qubit $\ket{\phi_i}$ which is either $\ket{0}$ or $\ket{1}$. Let us calculate the probability that $\ket{0}$ is transformed to $\ket{1}$. Denote the event by $\ket{0} \rightarrow \ket{1}$. Then:
% \begin{align*}
%     \Pr[\ket{0} \rightarrow \ket{1}] &= \Pr[\ket{0} \rightarrow \ket{1} | b = 0] \Pr[b = 0] + \Pr[\ket{0} \rightarrow \ket{1} | b = 1] \Pr[b = 1] \\
%     &= 0 \cdot (1 - \gamma) + \frac{1}{2} \cdot \gamma = \frac{\gamma}{2}.
% \end{align*}
% Let us also calculate the probability of the event $\ket{1} \rightarrow \ket{1}$:
% \begin{align*}
%     \Pr[\ket{1} \rightarrow \ket{1}] &= \Pr[\ket{1} \rightarrow \ket{1} | b = 0] \Pr[b = 0] + \Pr[\ket{1} \rightarrow \ket{1} | b = 1] \Pr[b = 1] \\
%     &= 1 \cdot (1 - \gamma) + \frac{1}{2} \cdot \gamma = 1 - \frac{\gamma}{2}.
% \end{align*}
% Let $p = 1 - \frac{\gamma}{2}$. 
Then the following operation on $\ket{x}$ exactly captures these transformations:
\[
\sqrt{p}I + \sqrt{\frac{1 - p}{\kappa-1}} X + \sqrt{\frac{1 - p}{\kappa-1}}X^2 + \cdots + \sqrt{\frac{1 - p}{\kappa-1}} X^{\kappa-1}.
\]
This is precisely the dit flip channel~\cite{Dutta_2023}. The whole circuit to implement the channel and then measure the outcome can be done as: 

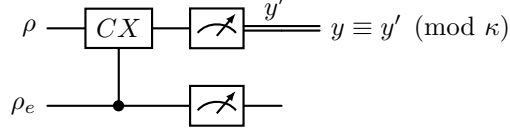
\begin{center}
\begin{quantikz}
\lstick{$\rho$} & \gate{CX} & \meter[label style={right = 0.5cm}]{y'} & \setwiretype{c} & \rstick{${y \equiv y' \pmod{\kappa}}$} \\ 
\lstick{$\rho_e$} & \ctrl{-1} & \meter{} & 
\end{quantikz}
\end{center}

where
\[
\rho_e = p \ket{0}\bra{0} + \frac{1 - p}{\kappa-1} \ket{1}\bra{1} + \cdots +  \frac{1 - p}{\kappa-1} \ket{\kappa-1}\bra{\kappa-1}.
\]
To see this, note that $\rho_e$ is a mixture state~\cite[\S 2.4.1]{nielsen2002quantum}. With probability $p$ it leaves the initial state unchanged. And with probability $(1-p)/(\kappa-1)$ it acts as a control for the $CX$ gate on $\rho$, where $X = \sum_{j = 0}^{d-1} \ket{j+1}\bra{j}$, and $X^{-1} = X^\dagger = \sum_{j = 0}^{d-1} \ket{j}\bra{j+1}$. This means that the state of the system after the $CX$ gate is:
\begin{align*}
    \rho' &= p\rho \otimes \ket{0}\bra{0} + \frac{(1-p)}{\kappa-1} X \rho X^{-1} \otimes \ket{1}\bra{1} + \\
    &\cdots + \frac{(1-p)}{\kappa-1} X^{\kappa-1} \rho X^{-(\kappa-1)} \otimes \ket{\kappa-1}\bra{\kappa-1}.
\end{align*}
Fix a measurement result ${j}$ of the computational basis, where $j \in \{0, 1, \ldots, d-1\}$. Then with probability $\text{tr} ((I \otimes \ket{j}\bra{j}) \rho')$, the state after measurement of the environment is:
\begin{align*}
    \rho'' = \frac{(I \otimes \ket{j}\bra{j}) \rho' (I \otimes \ket{j}\bra{j})}{\text{tr} ((I \otimes \ket{j}\bra{j}) \rho')}
\end{align*}
Initially we had $\rho = \ket{x} \bra{x}$, where $x \in \{0, 1, \ldots, \kappa-1\}$.
This implies that the outcome $j = 0$ occurs with probability $p$, leaving the state as $\rho'' = \frac{p \rho}{p} = \rho = \ket{x} \bra{x}$, whereas the outcomes $1 \le j \le \kappa-1$ each occur with probability $(1-p)/(\kappa-1)$, leaving the state as $\rho'' = X^j \rho X^{-j} = \ket{x + j}\bra{x + j}$, and any outcome $\kappa \le j \le d-1$ does not occur (probability 0). Now  measuring the state $\rho''$, again in the computational basis, gives the outcome $y' = x + j$ with $0 \le j \le \kappa-1$. Since $d > (\kappa-1)n$, and $n$ is at least $2$, we see that the maximum possible value of $x + j$, i.e., $2(\kappa-1)$, is less than $d$. Therefore, we can uniquely obtain $y \equiv y' \pmod{\kappa}$ as the correct outcome of the LDP mechanism. This value can then be used in Step 4 of the protocol.

\section{Fault-Tolerant Computation}
\label{sec:fault-tolerance}
In this section, we describe how the computation in the protocol can be done in a fault-tolerant manner. We begin with a review of stabilizer codes.  
% \begin{enumerate}
%     \item Fault-tolerant preparation of the initial GHZ state given in Eq.~\ref{eq:ghz}.
%     \item Fault-tolerant version of the circuit (shown in Figure~\ref{fig:teleport}) for teleporting a qudit from this state to a party. 
%     \item Fault-tolerant application of $Z$ gates and subsequent measurements by each party. 
% \end{enumerate}

\subsection{Stabilizer Codes}
\label{subsec:stabilizer-codes}
Quantum error-correction is most conveniently expressed in the stabilizer formalism. In particular, we will be using qudit stabilizer codes~\cite{Bullock_2007}. We now assume that $d$ is a prime. Note that this does not limit the application of our protocol, since our condition is only that $d > (\kappa - 1)n$, and hence any prime $d$ greater than this quantity can be chosen. An $[N, K]$-stabilizer code encodes $K$ logical qudits into $N$ physical qudits. We assume that the code can correct errors on up to $T$ encoded qudits. In the simplest case it can correct arbitrary errors on a single encoded qudit. Quantum errors can be modeled as a quantum operation which transforms a given state to another. Quantum operations themselves can be expressed as a sum of (Krauss) operators. A remarkable result in quantum-error correcting codes states that if an error-correcting code satisfies a discrete set of quantum-error correcting conditions with respect to the Krauss operators of the quantum error operation, then it can correct arbitrary errors. Since any operator can be written as a linear combination of Pauli matrices: $I$, $X$, $Z$ and $iXZ$, for $d = 2$, it follows that if an error-correcting code can correct these errors on say one qubit, it can correct arbitrary errors on that qubit. These results generalize to qudits with some modifications.

We consider the generalized Pauli group $\mathcal{P}^n_d$ of $n$-fold tensor products of the qudit operators $I$, $X$ and $Z$ (as defined earlier). For $k, \ell \in \mathbb{Z}_d$, these operators satisfy the commutation relation 
$X^k Z^l = \omega^{-kl}Z^l X^k$. For $n$-fold tensor product of these, we get: 
\begin{equation}
\label{eq:pauli-group}
  \mathcal{P}^n_d = \{\omega^jX^{\otimes \mathbf{k}}Z^{\otimes \mathbf{l}}; \mathbf{k},\mathbf{l}\in\mathbb{Z}_d^N , j\in \mathbb{Z}_d \}.  
\end{equation}
A qudit stabilizer group is a subgroup $S \subseteq \mathcal{P}^n_d$. The code subspace of $S$ consists of all vectors $\ket{\psi} \in (\mathbb{C}^d)^{\otimes N}$ which are stabilized by all elements of $S$, meaning, for all $s \in S$, we have $s \ket{\psi} = \ket{\psi}$. In other words, these belong to the $+1$ eigenstate of $s$. Let us denote this code space by $V_S$. For $V_S$ to be non-trivial we must have that $S$ is an Abelian group and $\omega^j I$ for $j \neq 0$ should not be in $S$~\cite{kawabata2023narain}. 
The subgroup $S$ can be succinctly represented by only $N-K$ generators, $g_1, \ldots, g_{N-K}$, which are themselves members of $\mathcal{P}^n_d$, and we write $S = \langle g_1, \ldots, g_{N-K}\rangle$. To express the dynamics of the logical (encoded) qudits, we have the logical operators $\bar{X}_1, \ldots, \bar{X}_K$ and $\bar{Z}_1, \ldots, \bar{Z}_K$, which act as the logical $X$ and $Z$ operators on the encoded qudits (one for each of the $K$ encoded qudits). These logical operators commute with the generators of $S$, and anti-commute with each other.  
An error operator $E$ is defined as an element of $\mathcal{P}^n_d$ which is a $T$-fold tensor product of $I$, $X$ and $Z$. In other words, it acts on at most $T$ qudits. The error $E$ is correctable if it anti-commutes with at least one generator of $S$. Suppose $g$ is one such generator. Then
\[
g E\ket{\psi} = \omega^j Eg\ket{\psi} = \omega^j E\ket{\psi}
\]
for some $j \neq 0$. This implies that the vector $E\ket{\psi}$ is in the $\omega^j$-eigenstate of $g$. If we measure $g$, the eigenvalue $\omega^j$ can thus be regarded as the syndrome of this error. Note that $g$ is a normal operator as it is a member of the Pauli group, and hence is an observable according to our definition in Section~\ref{sec:background}. This means that we can use it as a measurement operator. One can thus remove the error by applying $E^\dagger$ to $E\ket{\psi}$. If $U$ is any unitary operator acting on a qudit $\ket{\psi}$ of $V_S$, then we see that
\[
U \ket{\psi} = U g \ket{\psi} =UgU^\dagger U \ket{\psi},
\]
where $g$ is any generator of $S$. Thus $UgU^\dagger$ stabilizes $U\ket{\psi}$. It follows that the set $\{Ug U^\dagger\}$ for $g \in S$ is the set of generators of the codespace $UV_S$. $U$ is a member of the \emph{Clifford group} if for all $P \in \mathcal{P}^n_d$ we have $UPU^\dagger \in  \mathcal{P}_d$. Obviously $X$ and $Z$ are members of the Clifford group. Proposition~\ref{prop:ZX-interchange} shows that so is $H$. One can similarly show that the $CX$ gate also belongs to the Clifford group~\cite{gheorghiu2014standard}. Looking at all the circuits used in our protocol: Figures~\ref{fig:bell-state}, \ref{fig:ghz-state}, \ref{fig:teleport} and \ref{fig:local-operations}, we see that we only use gates from the Clifford group. Furthermore, initial states in our circuits are obtained only via tensor products of $\ket{0}$'s. Given an $[N, K]$-stabilizer code which corrects up to $T$ errors, we then have the following recipe to achieve fault-tolerance
\begin{enumerate}
    \item Fault-tolerant encoding of the initial state $\ket{0}^K$ to $\ket{0}^N$.
    \item Fault-tolerant versions of the (logical) Clifford gates $Z$, $X$, $H$ and $CX$ on the encoded qudits.
    \item Fault-tolerant measurement of the encoded state.
\end{enumerate}
All this can be done in the stabilizer framework. For the first step, we simply measure all the generators to identify any errors in the initialization, and correct them according to the syndrome as discussed above. For step (2), we simply update the set of generators $\{g_i\}$ as $\{Ug_iU^\dagger\}$, which remains in the Pauli group for Clifford gates. Any errors in the circuit can again be removed by measuring the generators. For the last step, we need to be able to implement computational basis measurements, i.e., $\bar{Z}$, in a fault-tolerant way. This can be done if the generators can be measured via a fault-tolerant procedure since $\bar{Z}$ is also part of the Clifford group. It turns out that the generators of the stabilizer code can be measured fault-tolerantly without disturbing the state~\cite{Bullock_2007}. As long as the errors are confined to at most $T$ qudits, the circuit will be fault-tolerant. In other words, we define a circuit to be fault-tolerant to $T$ failures if $T$ failures cause at most $T$ errors in an encoded block of qudits. Thus, we can achieve fault-tolerance if we have an example of such a stabilizer code. For more details in stabilzer codes, see~\cite[\S 10.5]{nielsen2002quantum}. 
Next we describe an example of a stabilizer code which can be used to implement these fault-tolerant operations. 

\subsection{Surface Codes}
\label{subsec:surface-codes}
%The initial generalized GHZ state can be prepared in a closed system and then it can be teleported quditwise to different parties. But teleportation of each qudit may not happen at the same time hence there is a possibility for logical errors in the meantime. So we need to encode the GHZ in quantum error correction codes to make it fault tolerant. 
One of the best examples of stabilizer codes to encode a large system is surface codes~\cite{bravyi1998sc, PhysRevA.86.032324}. The earlier surface codes for qubit systems have already been generalized for qudit systems, see for example~\cite{Bullock_2007, hussainanwar}. The surface code consists of an $L\times L$ lattice, each vertex of which represents a physical qudit. The surface code is stabilized by the stabilizer group $S = \langle A_s, B_p \rangle$ defined as: 
\begin{align}
\label{eq:sc-generators}
    A_s &= X_e \otimes X^{-1}_f \otimes X^{-1}_g \otimes X_h, \; \forall \text{ vertices } s \in V \nonumber\\
B_p &= Z_a \otimes Z_b \otimes Z^{-1}_c \otimes Z^{-1}_d, \; \forall \text{ faces }p\in P
\end{align}
where $e, f, g, h$ are edges surrounding a vertex $s\in V$ and $a, b, c, d$ are the edges surrounding a face (plaquette) $p \in P$, in the zigzag order $\lightning$. The lattice together with two example generators are shown in Figure~\ref{fig:surface-code} which is derived from a similar figure in~\cite{watson}. 

\begin{figure}
\centering
\begin{tikzpicture}
% grid
\draw (0.5,0) grid (4.5,4);
\draw[line width=3pt, line cap=round, dash pattern=on 0pt off 1cm](0.5,0.5) grid (4.5,4);
\draw[line width=3pt, line cap=round, dash pattern=on 0pt off 1cm](0.5,0) -- (4.5,0);

% faces
\node [fill=blue!20, draw=none, minimum size=0.98cm] 
      at (3.5,3.5) {};

\node[] at (3.5, 3.5) {$B_p$};

\node[draw=blue!80,circle,minimum size=0.2cm,inner sep=1pt, fill=blue!20] at (3.5,4) {\scriptsize $a$};
\node[draw=blue!80,circle,minimum size=0.2cm,inner sep=1pt, fill=blue!20] at (3,3.5) {\scriptsize $b$};
\node[draw=blue!80,circle,minimum size=0.2cm,inner sep=1pt, fill=blue!20] at (4,3.5) {\scriptsize $c$};
\node[draw=blue!80,circle,minimum size=0.2cm,inner sep=1pt, fill=blue!20] at (3.5,3) {\scriptsize $d$};

% vertices

\node [fill=red!20, rotate=45, draw=none, minimum size=0.8cm] 
      at (4,2) {};

\node[] at (4, 2) {$A_s$};

\node[draw=red!80,circle,minimum size=0.2cm,inner sep=1pt, fill=red!20] at (4,2.5) {\scriptsize $e$};
\node[draw=red!80,circle,minimum size=0.2cm,inner sep=0.7pt, fill=red!20] at (3.5,2) {\scriptsize $f$};
\node[draw=red!80,circle,minimum size=0.2cm,inner sep=1pt, fill=red!20] at (4.5,2) {\scriptsize $g$};
\node[draw=red!80,circle,minimum size=0.2cm,inner sep=1pt, fill=red!20] at (4,1.5) {\scriptsize $h$};

% logical X

\draw[dashed, red!80] (1.5,4) -- (1.5,0);

\node[draw=red!80,circle,minimum size=0.2cm,inner sep=1pt, fill=red!20] at (1.5,4) {\scriptsize $j$};

\node[draw=red!80,circle, minimum size=0.2cm,inner sep=1pt, fill=red!20] at (1.5,3) {\scriptsize $j$};

\node[draw=red!80,circle,minimum size=0.2cm,inner sep=1pt, fill=red!20] at (1.5,2) {\scriptsize $j$};

%\node[draw=red!80,circle,minimum size=0.2cm,inner sep=1pt, fill=red!20] at (1.5,1) {\scriptsize $j$};

\node[draw=red!80,circle,minimum size=0.2cm,inner sep=1pt, fill=red!20] at (1.5,0) {\scriptsize $j$};

\node[] at (5, 1) {$\bar{Z}^k$};

% logical Z

\draw[dashed, blue!80] (0.5,1) -- (4.5,1);

\node[draw=blue!80,circle,minimum size=0.2cm,inner sep=1pt, fill=blue!20] at (0.5,1) {\scriptsize $k$};

%\node[draw=blue!80,circle,minimum size=0.2cm,inner sep=1pt, fill=blue!20] at (1.5,1) {\scriptsize $k$};

\draw[red!80, fill=red!20] (1.5,1) + (0, 0.2) arc (90:270:0.2cm);

\draw[blue!80, fill=blue!20] (1.5,1) + (0, -0.2) arc (270:450:0.2cm);

\node[] at (1.5,1) {\scriptsize $j \; k$};

\node[draw=blue!80,circle,minimum size=0.2cm,inner sep=1pt, fill=blue!20] at (2.5,1) {\scriptsize $k$};

\node[draw=blue!80,circle,minimum size=0.2cm,inner sep=1pt, fill=blue!20] at (3.5,1) {\scriptsize $k$};

\node[draw=blue!80,circle,minimum size=0.2cm,inner sep=1pt, fill=blue!20] at (4.5,1) {\scriptsize $k$};

\node[] at (1.5, -0.5) {$\bar{X}^j$};

\end{tikzpicture}
 \caption{An example of a surface code derived from~\cite{watson}. Qudits are black dots. The face labelled $B_p$ is an example stabilizer generator with the $Z$ and $Z^{-1}$ operators shown in Eq.~\eqref{eq:sc-generators}. The ``diamond'' labelled $A_s$ is an example stabilizer generator with the $X$ and $X^{-1}$ operators shown in Eq.~\eqref{eq:sc-generators}. The blue dots are for face operators $Z$ and red dots for vertex operators $X$. Logical $\bar{X}^j$ operation is indicated by applying $X^j$ to each of the qudit in the highlighted vertical line (indicated by $j$). The logical $\bar{Z}_k$ operation is shown likewise on the horizontal line (indicated by $k$).}
\label{fig:surface-code}
\end{figure}
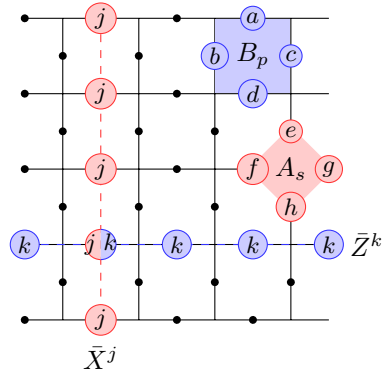

%Let us consider the Pauli group $\mathcal{P}_d$ generated by $\langle X,Z \rangle$
%\[
%X = \sum_{j \in \mathbb{Z}_d} %\ket{j\oplus1}\bra{j},\;
%Z = \sum_{j \in \mathbb{Z}_d} \omega^j %\ket{j}\bra{j}
%\]
%where $X$ and $Z$ are the generalized operators for qudit. For $d>2$ they are not Hermitian anymore but they still have the eigenvalues $\omega^k$ for $k\in \mathbb{Z}_d$. For some $k$,$l \in j \in \mathbb{Z}_d$ it has the commutation relation as: $X^k Z^l = \omega^{-kl}Z^l X^k$. Thus for $n$ qudit system the Pauli group will be:
%\[
%\mathcal{P}_d^n = \{\omega^jX^{\otimes k}Z^{\otimes l}; k,l\in\mathbb{Z}_d^n , j\in \mathbb{Z}_d \}
%\]
%The stabilizer group is subgroup of the
%se Pauli group defined as: $S \subseteq \mathcal{P}_d$ such that the code subspace of any stabilizer $s\in S$ is the joint $+1$ eigenspace. 

For all $s, p \in V, P$, these generators commute. To see this, note that they trivially commute if they do not have any vertices in common (the tensor product is over different qudits). Otherwise, they can have exactly two vertices in common. For instance. Let us assume that $s$ is the vertex with top and left edges $c$ and $d$ common with $B_p$ in the figure. Let us call the remaining edges of $s$ as $e$ and $f$. Then:
\begin{align*}
   A_sB_p &= (X_c \otimes X^{-1}_d \otimes X^{-1}_e \otimes X_f)(Z_a \otimes Z_b \otimes Z^{-1}_c \otimes Z^{-1}_d) \\
   &= Z_a \otimes Z_b \otimes X_cZ_c^{-1} \otimes X_d^{-1}Z_d^{-1} \otimes X_e^{-1} \otimes X_f\\
   &=  Z_a \otimes Z_b \otimes \omega^{-(1)(-1)} Z_c^{-1} X_c \otimes \omega^{-(-1)(-1)} Z_d^{-1}X_d^{-1} \otimes X_e^{-1} \otimes X_f\\
   &= \omega^0 Z_a \otimes Z_b \otimes  Z_c^{-1} X_c \otimes  Z_d^{-1}X_d^{-1} \otimes X_e^{-1} \otimes X_f\\
   &= (Z_a \otimes Z_b \otimes Z^{-1}_c \otimes Z^{-1}_d) (X_c \otimes X^{-1}_d \otimes X^{-1}_e \otimes X_f)\\
   & = B_p A_s
\end{align*}

The two logical operators $\bar{X}^j$ and $\bar{Z}^k$ are also shown in Figure~\ref{fig:surface-code}, which apply the $X^j$ and $Z^k$ operators on each of the qudits along the vertical and horizontal lines, respectively. These can easily be seen to anti-commute with one another, as they have only one vertex in common. Furthermore, they commute with all the generators. The generator $A_s$ (respectively $B_p$) trivially commutes with $\bar{X}_j$ (respectively $\bar{Z}_k$),  and each generator $A_s$ (respectively, $B_p$) has two vertices in common with $\bar{Z}^k$ (respectively, $\bar{X}_j$). This lattice is used to encode one qudit. This completes the description of the surface code as a stabilizer code. Thus, we can use it to encode each qudit in our protocol, and then implement the encoded versions of the gates $Z$, $X$, $H$ and $CX$ as discussed in Section~\ref{subsec:stabilizer-codes}. In Appendix~\ref{app:lattice-surgery} we briefly describe another quantum error correcting method based on lattice surgery. 

\descr{Noise Model and Decoding.} For completeness, we can consider the \emph{independent noise model} from~\cite{watson}. 
In their noise model $X^k$ and $Z^k$ errors can occur to a data qudit with probability $p/(d-1)$ independently for a parameter $p$ and $1\leq k\leq (d-1)$. The error threshold is an upper bound such that any quantum error correction code with error probability per component below this threshold decreases the error rate on the encoded qudits. The error threshold defined by~\cite{watson}, denoted $p_{\text{th}}^d$, for their decoding algorithm increases with the dimension $d$ and saturates around $8.3\%$ in the case of an $L\times L$ lattice. This means the logical error rate will decrease with increasing code distance if the probability of physical error rate is below $8.3\%$. 

%In our system we are going to use the decoder algorithm described in \cite{watson}. In this paper the authors proved that for a noise model where an $X^k$ and $Z^k$ errors can occur to a data qudit with probability $p/(d-1)$ independently for a parameter $p$ and $1\leq k\leq (d-1)$, the error threshold ($p_{th}^d$ increases with the dimension $d$ and saturates around $8.3\%$ in the case of an $L\times L$ lattice. It means the logical error rate will decrease with increasing code distance if the probability of physical error rate is below $8.3\%$. 

\descr{Bell State Channels.} An aspect missing in the discussion above is the creation of quantum channel through the generalized Bell pairs (see Figure~\ref{fig:channels}). This can be done fault-tolerantly through \emph{quantum repeaters}~\cite{RevModPhys.95.045006, 7010905, Rohde_2021}. Quantum repeaters segment a network link, apply entanglement between segments and connect them to achieve long-range end-point entanglement. 

%The whole lattice will represent or encode one logical qudit. Similarly like qubit surface codes there are logical operators defined here for tackling additional degree of freedom. These logical operators are logical $\bar{X}$ that connects two opposite smooth edges and logical $\bar{Z}$ connecting two opposite rough edges. For illustration we will use the \textit{FIG 2} of Watson et al. 2015 \cite{watson} as Figure~\ref{fig:surface-code}.

% \begin{figure}
% \begin{center}
%     \includegraphics[width=0.5\textwidth]{origin_image.pdf}
% \end{center}

%     \caption{"An example of a distance 5 surface code.
% Qudits are shown as black dots, arranged on the edges of a lattice
% with two types of boundary: rough and smooth. For clarity, when
% an arbitrary $X^j$ or $Z^k$ Pauli operator acts on a physical qudit, we
% only include the exponents $j$ and $k$ on the edges of the figure. We
% use red for $X^j$ errors and vertex operators, and blue for $Z^k$ errors or
% plaquette operators. (a) and (b) An example of a single plaquette and
% vertex operator, respectively. (c) An example of a deformed rough
% edge plaquette operator (three-body operator). Note that the vertex
% operators are deformed at smooth edges. (d) and (e) An example of a
% pair of anticommuting logical operators." \cite{watson}}
% \label{fig:expsurf}
% \end{figure}

\section{Physical Realization of Our Protocol}
\label{sec:physical}

\descr{Experimental Systems.} Some of the essential ingredients of our protocol have already been achieved experimentally. Anonymous broadcasting has been demonstrated in an eight node network using photon polarization entangled qubits \cite{Huang2022ExpAB}, and deterministic synthesis of multi-partite entangled states of higher dimensional systems has been shown \cite{edmunds2024constructingspin1haldanephase}.
We note that an alternative to using high dimensional quantum spins is to use continuous variable systems like temporal-\cite{Men2011Temp} or frequency-\cite{Wang2014Freq} modes of light. In~\cite{Men2018anon}, it was shown that anonymous broadcasting can be performed using a continuous variable surface code as a resource. Such a state can be prepared using Gaussian operations where the limit on the effective local dimension for each party is in principle unbounded but in practice is constrained by the amount of squeezing available in the state preparation. There, a real valued message can be sent anonymously with a channel capacity given by $C=\frac{1}{2}\log(1+\alpha)$ where the signal-to-noise ratio is $\alpha=2\tau^2 s^2$ with $\tau^2$ being the variance of the message to be broadcast and $s$ the squeezing parameter required for the state preparation. While this version of the protocol is not strictly fault tolerant, it can be done in an error mitigated way by allowing the parties access to local bosonic reservoirs that continously cool the state close to the code subspace \cite{Men2018anon}.   

\descr{Qudit Systems.} Compared to a qubit system, a qudit system has many advantages. One such advantage is a much larger state space. This has for instance enabled us to propose our protocol for $\kappa$-ary randomized response for any $\kappa \ge 2$, as we can accomodate the sum of all inputs from clients as long as they are less than $d$, which can be chosen as a large enough prime.  This is also one of the reasons that earlier works on anonymous broadcast in a quantum setting built on qubit systems were limited to sending bits~\cite{qa-transmissions}. Other advantages include simpler circuits and more efficient algorithms~\cite{wang2020qudit-systems}. It is no wonder then that considerable effort is being put in place to make qudit systems physically realizable. 

Many physical systems that are used to implement the qubit system already have more than two states, such as the frequency of photons, which can be utilized for qudit systems~\cite{wang2020qudit-systems}. Recent results have successfully reconstructed density matrices for up to $d = 8$ using biphoton frequency combs~\cite{lu2022bayesian}. Physicists have also been able to create a two-qudit entangled gate for up to $d = 5$ using trapped ions~\cite{hrmo2023native}. The work in~\cite{erhard2020advances} experimentally demonstrates a proof-of-concept qudit-based quantum processor for $d = 4$ using photonic systems. These and many other experimental efforts indicate that higher-dimensional qudit systems will be a reality within a timeframe not too distant from the realization of high-dimensional qubit systems. 

\section{Related Work}
\label{sec:rw}
Our protocol bears resemblance to the e-voting protocol from~\cite{hillery2006voting} as well as the protocol from~\cite{qa-transmissions} for anonymously broadcasting a single bit to a group of users. The e-voting protocol~\cite{hillery2006voting} considers the binary case, i.e., $\kappa = 2$, each voter either sends 0 or 1 as his/her vote. The aggregator needs to count the number of yes votes. The protocol does not include differential privacy, either local or in the shuffle model. Furthermore, they do not consider how their protocol can be implemented fault tolerantly or how the initial state can be transported to the voters. The protocol from~\cite{qa-transmissions} is for sending a broadcast bit anonymously to a group of users. Their setting is different to ours as there is no central aggregator. Furthermore, their target is to solve the dining cryptographer's problem in the quantum setting, and hence does not involve differential privacy. For multiple senders, they introduce a more complicated protocol involving collision detection. In our case, we use qudits to resolve the issue of multiple senders and non-binary messages. Once again, being an earlier work, they do not consider fault tolerant implementation of their protocol. A more recent protocol~\cite{lipinska2018wstate} looks at sending a quantum message from a sender to a receiver in a manner that their identities remain anonymous to the network of $N$ nodes. To achieve this they use the so-called $W$ state: $\frac{1}{\sqrt{N}}\sum_{j=1}^N\ket{j}$, where $j \in \mathbb{Z}_N$ is the binary vector with the solitary one in the $j$th location. Once again their setting, goal and protocol is different to ours. 

There have been a number of works that discuss differential privacy in the quantum setting. They primarily focus on definitional aspects of applying differential privacy to quantum information, and its properties as well as realizing differentially private noise via noisy quantum channels~\cite{zhou2017dp-quantum, hirche2023quantum, aaronson2019gentle}. In our work, we only apply differential privacy to classical information, i.e., states that can be represented as $\rho = \ket{j}\bra{j}$, where $\ket{j}$ is a qudit, whereas these works consider more general quantum information, e.g., a superposition of qudits. For classical information these definitions are equivalent to classical differential privacy (Section~\ref{subsec:dp}). The definition in~\cite{zhou2017dp-quantum} and \cite{hirche2023quantum} defines quantum differenital privacy in terms of trace distance between two quantum states $\rho$ and $\sigma$, given as $\frac{1}{2}\text{tr}(\sqrt{(\rho - \sigma) (\rho - \sigma)^\dagger})$. Denoting $\rho = \ket{j}\bra{j}$ and $\sigma = \ket{k}\bra{k}$, we see that $\sqrt{(\rho - \sigma) (\rho - \sigma)^\dagger} = \sqrt{\ket{j}\bra{j} + \ket{k}\bra{k}} = \ket{j}\bra{j} + \ket{k}\bra{k}$, so that the trace distance becomes 1. Hence the resulting definition is equivalent to the classical definition. The definition from~\cite{aaronson2019gentle} defines two quantum states $\rho$ and $\sigma$ as being neighbors if we can reach $\rho$ from $\sigma$ or $\sigma$ from $\rho$ with a quantum operation on a single register only (in our case, a single qudit). With $\rho$ and $\sigma$ as defined above, we see that we have $X^\ell \rho X^{-\ell} = \sigma$ and $X^{-\ell} \sigma X^\ell = \rho$, where $\ell \equiv k - j \pmod{d}$. This satisfies the definition of a general quantum operation in~\cite{aaronson2019gentle}, and hence the differential privacy definition is again equivalent. It is due to this reason that we do not give a separate definition of quantum differential privacy in our paper. 

Lastly, there are a growing number of works exploring applications of differentially private algorithms and protocols in the quantum setting. The work in~\cite{guan2023violations} proposes methods to detect violations of differential privacy for quantum algorithms. These methods provide a counterexample of a pair of quantum states which breach privacy, revealing the cause of the breach. Such empirical methods are important for practical applications of differential privacy. Researchers have also looked at making quantum machine learning differentially private~\cite{watkins2023qml-dp}. However, this work looks at running a quantum machine learning algorithm over a central dataset, as opposed to our distributed case.

\section{Conclusion and Future Directions}
\label{sec:conclude}
We have presented a quantum protocol for differential privacy in the shuffle model for the $\kappa$-ary randomized response algorithm. The key advantage of the quantum approach is that we can implement the shuffle using properties of quantum entanglement without requiring additional mechanisms and trust assumptions to implement the shuffle as in the classical setting. A key feature of our protocol is that it can be implemented only using Clifford gates, which are easy to implement using quantum error correction schemes, and make the protocol highly efficient as they can be simulated efficiently using a classical computer. We therefore describe how our protocol can be implemented fault-tolerantly, which is suited to the current noisy intermediate-scale quantum (NISQ) era. There are a number of ways in which the current protocol can be improved. One direction is to consider weaker threat models, where clients could collude and/or be malicious. We note that such threat models are quite challenging even for differential privacy in the classical setting. Another direction is to expand the protocol to consider more general quantum states rather than classical states as is considered in our paper. A third direction is to construct protocols for local differential privacy other than randomized response, such as the Laplace mechanism, and for more sophisticated aggregate functions other than summation.

\section*{Acknowledgments}
This work was funded in part by a Future Communications Research Centre grant from Macquarie University. G.K.B. acknowledges support from the Australian Research Council Centre of Excellence for Engineered Quantum Sys- tems (Grant No. CE 170100009). 

\bibliography{quantum-ref.bib}

\appendix

\section{De-Biased Sum}
\label{app:de-bias}
Let $X_i$ denote the random variable representing user $i$'s output after running Algorithm~\ref{algo:bounded-noise}. Let $X = \sum_i^n X_i$. We are interested in:
\[
    \mathbb{E}(X) = \sum_{i=1}^n \mathbb{E}(X_i)
\]
Let $p_j$ be the probability that user $i$ outputs $j \in \{0, 1, \ldots, \kappa-1\}$. Let $q_j$ be the true probability of any user having input $j$. Then,
\begin{align*}
    p_j &= \left( 1 - \gamma + \frac{\gamma}{\kappa}\right) q_j + \frac{\gamma}{\kappa} (1 - q_j)\\
    &= (1 - \gamma)q_j + \frac{\gamma}{\kappa}
\end{align*}
Then
\begin{align*}
    \mathbb{E}(X_i) &= \sum_{j = 0}^{\kappa-1} j p_j \\
    &=\sum_{j = 0}^{\kappa-1} j \left( (1 - \gamma)q_j + \frac{\gamma}{\kappa} \right)\\
    &= (1 - \gamma) \left( \sum_{j = 0}^{\kappa-1} j q_j \right) + \frac{\gamma (\kappa-1)}{2} \\
    &= (1 - \gamma) \mu + \frac{\gamma (\kappa-1)}{2}
\end{align*}
where $\mu = \sum_{j = 0}^{\kappa-1} j q_j$ is the true expected input of any user. Thus,
\begin{align*}
    \mathbb{E}(X) &= n \left( (1 - \gamma) \mu + \frac{\gamma (\kappa-1)}{2}\right)\\
    \Rightarrow {n\mu} &= \frac{1}{1 - \gamma}\left( {\mathbb{E}(X)} - \frac{\gamma (\kappa-1)n}{2} \right).
\end{align*}
Therefore, the expected value of the sum output by the LDP algorithm, i.e., $\mathbb{E}(X)$, gives us the expectation of the sum of true inputs, i.e., $n\mu$. Thus, given the sum of these values for a sample, we can estimate the true sum as above.

\section{Lattice Surgery}
\label{app:lattice-surgery}
Lattice surgery is another method to implement quantum error correcting code. We can use it to make a fault-tolerant GHZ state. Just like lattice surgery for qubit system \cite{Horsman_2012, fowler2019low, Litinski2019gameofsurfacecodes} the framework's merging, splitting and other type of `surgeries' have been conceptualized by Cowtan~\cite{cowtan2022qudit}. We need the splitting operation (more precisely, smooth spitting) to make a logical GHZ state with an encoded logical qudit state. For smooth splitting the qudits of an intermediate row parallel to the smooth boundaries  are measured out in the $X$ basis. 
%Here $\ket{\delta_i}$ is the Fourier transformed image of $\ket{i}$ (note for qubit smooth splitting the intermediate qubits are measured in X basis which is the same if we see the $\{\ket{\delta_i}\}$ for $d=2$ as it is described in \cite{Horsman_2012}). 
The results of \cite{Horsman_2012} shows that after smooth spitting in this basis, the new born two surface will represent two different logical qudits with the same logical $\bar{Z}$ operator for both but different logical $\bar{X}$ operators and the two logical qudits will be entangled. More presisely:
\[
\ket{i}_L \rightarrow \ket{i}_L \otimes \ket{i}_L
\]
Hence if we encode a state $\ket{+} = \frac{1}{\sqrt{d}}\sum_{i\in\mathbb{Z}}\ket{i}$, $+1$ eigenstate of logical $\bar{X}$ in the data qudits of the surface patch after one smooth split we will get a logical Bell pair and so on:
\[
\ket{+}_L \rightarrow \frac{1}{\sqrt{d}}\sum_{i\in\mathbb{Z}}\ket{ii} \rightarrow \frac{1}{\sqrt{d}}\sum_{i\in\mathbb{Z}}\ket{iii} \rightarrow \cdots \rightarrow \frac{1}{\sqrt{d}}\sum_{i\in\mathbb{Z}}{\ket{i}}^{\otimes n}
\]
A diagram for this kind of operation is given in Figure~\ref{fig:surg}. 

\begin{figure}
%\begin{center}
\centering
\begin{tikzpicture}
% grid
\draw (0,0.5) grid (8,2.5);
\draw[line width=3pt, line cap=round, dash pattern=on 0pt off 1cm](0.5,0.5) grid (8,2);
\draw[line width=3pt, line cap=round, dash pattern=on 0pt off 1cm](0,2.5) -- (8,2.5);
\draw[line width=3pt, line cap=round] (0, 0.5) -- (0, 0.5);
\draw[line width=3pt, line cap=round] (0, 1.5) -- (0, 1.5);

% measured out qudits

\node[draw=orange!80,circle,minimum size=0.2cm,inner sep=1pt, fill=orange!20] at (2.5,2) {};
\node[draw=orange!80,circle,minimum size=0.2cm,inner sep=1pt, fill=orange!20] at (2.5,1) {};

\node[draw=orange!80,circle,minimum size=0.2cm,inner sep=1pt, fill=orange!20] at (5.5,2) {};
\node[draw=orange!80,circle,minimum size=0.2cm,inner sep=1pt, fill=orange!20] at (5.5,1) {};

\end{tikzpicture}
%    \includegraphics[width=0.5\textwidth]{surgery.drawio.pdf}
%\end{center}

    \caption{Construction of a logical GHZ state with 3 qudits using lattice smooth splitting. Here every dot represents a data qudit and the orange dots represent measured out qudits}
    \label{fig:surg}
\end{figure}
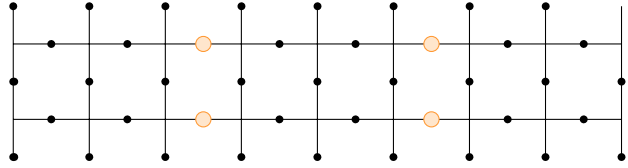

In \cite{Horsman_2012} we see that the threshold limit increases with the degree of the qudit. So it gives us a chance to work with more number of qudits since number of qudits must be less than the dimension. This decoder can work for error correction of finite dimension qudits. For example they have shown it gives good threshold for $d=7919$. Since $n\leq d$ and with increasing $L$ the rate of success is higher for error correction, we can have a high number of $L$ and $d$. Then we can encode the $\ket{+}$ state into that $L\times L$ surface code and make it a fault tolerant logical $\ket{+}$. Then we can do lattice splitting to achieve our logical GHZ state. This GHZ state will be fault tolerant since it is assured by the fault tolerant nature of lattice surgery. Hence this way we can achieve a fault tolerant initial state for our protocol with $L^2$ number of physical qudits. 

We also need to see how to perform teleportation in a fault-tolerant way. We have already sketched fault-tolerant construction of the GHZ state. The server can encode his/her share of the qudit from each Bell pair in the surface code with the state injection method~\cite{Horsman_2012}. So if we can have fault tolerant construction of generalized $CX$ and Hadamard gate, we can perform teleportation demonstrated in Figure~\ref{fig:teleport} with logical qudits and gates. Fortunately we can construct these gates with the help of lattice surgery as shown in \cite{cowtan2022qudit}. To construct fault-tolerant $CX$, we will perform a smooth split on each qudits of the patches of logical GHZ. Then if we perform a rough merge (performing merge operation along the smooth surface of surface codes) with the server owned logical qudit of the respective Bell pair, we will get a logical $CX$. For the fault tolerant $H$ gate we are going to use the antipode operation, where the surface will be rotated such that the vertex and face operations will be exchanged i.e., the previous $X$ and $X^{-1}$ will be now be replaced with $Z$ and $Z^{-1}$ which will give us the logical Hadamard operation or Fourier transform. Thus, we have our logical $CX$, $H$ and $GHZ$, we already have fault=tolerant Bell pairs, we can have the error corrected $\ell$ and $s$ measurements (Figure~\ref{fig:teleport}) and thus the clients will have their non-faulty states via teleportation. 

\end{document}